\documentclass[journal]{IEEEtran}%max 13 pages, minimo 10 pt
\PassOptionsToPackage{bookmarks={false}}{hyperref}

\usepackage[latin9]{inputenc}
\usepackage{color}
\usepackage{amsmath}
\usepackage{amsthm}
\usepackage{amssymb}
\usepackage{graphicx}
\usepackage[T1]{fontenc}
\usepackage[active]{srcltx}
\usepackage{color}
\usepackage{float}
\usepackage{setspace}
\usepackage{enumitem}
\usepackage{subfloat}

\usepackage{url}

\usepackage{epsfig}
\usepackage[noadjust]{cite}
\usepackage{verbatim}
\usepackage{balance}
\usepackage{subfigure}
\usepackage{graphicx}

\newcommand{\beq}{\begin{equation}}
\newcommand{\eeq}{\end{equation}}

\def\bx{\mbox{\boldmath $x$}}

\def\bY{\mathbf{Y}}
\def\bD{\mbox{\boldmath $D$}}
\def\bH{\mbox{\boldmath $H$}}

\def\bU{\mbox{\boldmath $U$}}

\def\bQ{\mathbf{Q}}

\def\bZ{\mathbf Z}

\newcommand{\ds}{\displaystyle}

\newtheorem{theorem}{Theorem}

\newtheorem{lemma}[theorem]{Lemma}

\newtheorem{algorithm}{Algorithm}
\floatstyle{ruled}
\newfloat{algorithm}{tbp}{loa}
\providecommand{\algorithmname}{Algorithm}
\floatname{algorithm}{\protect\algorithmname}

\begin{document}
\title{\vspace{-0.5cm} Joint Optimization of Radio and Computational Resources
for Multicell  Mobile-Edge Computing \vspace{-0.1cm}}
\author{Stefania Sardellitti,  Gesualdo Scutari,  and Sergio Barbarossa\vspace{-0.7cm}
\thanks{S. Sardellitti and S. Barbarossa are with the Dept. of Information Engineering, Electronics and Telecommunications,
 Sapienza University of Rome,  Rome, Italy. Emails:  \texttt{<stefania.sardellitti,} \texttt{sergio.barbarossa>@uniroma1.it}.\newline G. Scutari is with the  Dept. of Electrical Engineering, State University of New York at Buffalo, Buffalo, USA. Email: \texttt{gesualdo@buffalo.edu.}\newline
\indent The work of  Barbarossa and Sardellitti  was supported by the European Community 7th Framework Programme Project ICT-TROPIC, under grant nr. 318784.
The work of Scutari  was supported by the USA NSF under
Grants CMS 1218717 and CAREER Award No. 1254739.
 Part of this work was presented at IEEE SPAWC 2014  \cite{SarScuBarSPAWC14} and at IEEE CloudNet 2014 \cite{Sard_Cloud_net}. }}% <-this

\maketitle

\begin{abstract}
Migrating computational intensive tasks from mobile devices to more resourceful cloud servers is a
promising technique  to increase  the computational capacity of mobile devices while saving
their battery energy.
%Compared to desktop computers, mobile devices
%possess less computing power, memory, storage capacity, and network bandwidth.
%Bringing the radio and computational resources closer to the mobile user can significatively reduce their power
%consumption and affect the communication delay by improving the users interactive response.
In this paper, we  consider a MIMO multicell system where multiple mobile users (MUs) ask for computation offloading to  a common  cloud server. We formulate the  offloading problem as the \emph{joint} optimization of the radio resources$-$the transmit precoding matrices of the MUs$-$and the  computational resources$-$the CPU cycles/second assigned  by the cloud to each  MU$-$in order to minimize the overall  users' energy consumption, while meeting latency constraints. The resulting optimization problem is nonconvex (in the objective function and constraints). Nevertheless,  in the single-user case,  we are able to express  the global optimal  solution in closed form.  In the more challenging multiuser scenario,  we propose   an iterative   algorithm, based on a novel successive convex approximation technique, converging  to a local optimal solution of the original nonconvex problem.  Then, we reformulate the algorithm in a distributed and parallel implementation across the radio access points, requiring only a limited coordination/signaling with the cloud. Numerical results show that the proposed schemes outperform disjoint optimization algorithms. % based be implemented in a centralized or parallel manner  and converging to a local optimal solution of the original nonconvex problem.\vspace{-0.2cm}
  %The algorithm is also suitable for a parallel implementation across the access point, with limited coordination/signaling with the cloud.
\end{abstract}\vspace{-0.2cm}


\section{Introduction}
Mobile terminals, such as smartphones, tablets and netbooks, are increasingly penetrating into our everyday lives as convenient tools for communication, entertainment, business, social networking, news, etc. Current predictions foresee
a doubling of mobile data traffic every year.
 %Thanks to the rapid Internet environment development, mobile phones are foreseen as the most common web access
 %entities in the next years. As predicted by Cisco \cite{Cisco},  monthly global data mobile traffic will surpass
% $10$ exabytes in 2017.
 %Current mobile devices have many advantages such as mobility, communication and sensing capability.
 %The advancement in computing, storage and network technologies has endowed them with a diversity of applications comparable to those in traditional
% desktop PCs.
However such a growth in mobile wireless traffic is not matched with an equally fast improvement on mobile handsets'  batteries, as testified in \cite{Palacin}. The  limited battery lifetime
is then going to represent the stumbling block  to the deployment of computation-intensive  applications for mobile devices.
% Despite the fast developments of key components such as CPU, memory, and wireless access technologies, %it remains still a challenge to
 % run complex applications and storage operations on mobile devices because of their limited power, computation and storage capabilities.
 %Since battery capacity has seen a slow improvement  in the last decade
 %\cite{Robinson},\cite{Palacin},
 % one of the major impediment to run sophisticated applications on mobiles remains their limited battery lifetime.
% a major impediment to run sophisticated applications on mobiles remains their limited battery lifetime.
At the same time, in the Internet-of-Things (IoT) paradigm, a myriad of heterogeneous devices, with a wide range of
computational capabilities, are going to be interconnected. For many of them, the local computation resources are insufficient to run sophisticated applications.
In all these cases, a  possible strategy to overcome the above energy/computation bottleneck consists in enabling  resource-constrained mobile devices to offload their most energy-consuming tasks to nearby more resourceful servers.
This strategy has a long history and is reported in the literature
under different names, such  as {\it cyber foraging} \cite{Sharifi-Kafaie-Kashefi}, or {\it computation offloading} \cite{Kumar-Liu-Lu-Bhargava}. In recent years, cloud computing (CC) has provided a strong impulse to computation offloading through virtualization, which decouples the application environment from the underlying hardware resources
and thus
%
%Thanks to virtualization, each application launched by a mobile device is associated to a cloud clone, which runs on a virtual machine (VM) with guaranteed isolation and protection of programs and data.
%One of the key features of CC is that  the number of virtual machines can scale on demand. This
%
enables an efficient usage of available computing resources.
%Coupling mobile radio access to cloud computing gives rise to
In particular, Mobile Cloud Computing (MCC) \cite{Fernando_et_al}  makes possible for mobile users to access cloud resources, such as   infrastructures, platforms, and software, on-demand.
Several works addressed mobile computation offloading, such as  \cite{Yang,Wolski,Zhang,Kaushi,Cardellini,Kumar_Lu, Cui, Miettinen,Huang-Wang-Niyato_TW12,Kao:Globecom:2014}.
Recent surveys are \cite{Fernando_et_al}, \cite{liu2013gearing},  and \cite{Sanaei-14}. Some works addressed the problem of program partitioning and offloading the most demanding program tasks, as e.g. in \cite{Yang,Wolski,Zhang,Kaushi}. Specific examples of mobile computation offloading techniques are: \textit{MAUI} \cite{Maui},  \textit{ThinkAir}  \cite{Kosta_2012}, and \textit{Phone2Cloud}  \cite{Phone2Cloud}.
The trade-off between the energy spent for computation and communication was studied in   \cite{Kumar_Lu, Cui, Miettinen, Wen}. A dynamic formulation of computation offloading was proposed in \cite{Huang-Wang-Niyato_TW12}. These works optimized offloading strategies, assuming a given radio access, and concentrated on single-user scenarios. In \cite{Barbarossa_FuNeMS2013}, it was proposed  a {\it joint} optimization of radio and computational resources, for the single user case.
The joint optimization was then extended to the multiuser case in \cite{Barbarossa_SPAWC2013}; see also \cite{Barbarossa_SPM2014} for a recent survey on joint optimization for computation offloading in a 5G perspective. The optimal joint allocation of radio and computing resources in \cite{Barbarossa_SPAWC2013}, \cite{Barbarossa_SPM2014} was assumed to be managed in a centralized way in the cloud. A decentralized solution, based on a game-theoretic formulation of the problem, was recently proposed in \cite{Chen_2014}, \cite{Cardellini}.
In current cellular networks, the major obstacles limiting an effective deployment of MCC strategies: i) the energy spent by mobile terminals, especially cell edge users, for radio access; and ii) the latency  experienced in reaching the (remote) cloud server through  a  wide area network (WAN). Indeed, in macro-cellular systems, the transmit power necessary for cell edge users to access a remote base station might null all potential benefits coming from offloading. Moreover, in many real-time mobile applications (e.g., online games, speech recognition, Facetime) the user Quality of Experience (QoE) is strongly affected by the system response time. Since controlling latency over a WAN might be very difficult, in many circumstances the QoE associated to MCC could be poor.

A possible way to tackle these challenges is  to bring {\it both} radio access and computational resources  closer to MUs.
This idea was suggested in
\cite{Satyanarayanan-Bahl-Caceres-Davies,liu2013gearing}, with the introduction of \textit{cloudlets}, providing proximity radio access to fixed servers through Wi-Fi. However, the lack of available fixed servers could limit the applicability of cloudlets. The European project TROPIC \cite{TROPIC} suggested to endow small cell LTE base stations with, albeit limited, cloud functionalities. In this way, one can exploit the potential dense deployment of small cell base stations to facilitate proximity access to computing resources and have advantages over Wi-Fi access in terms of Quality-of-Service  guarantee and a single technology system (no need for the MUs to switch between cellular and Wi-Fi standards). Very recently, the European Telecommunications Standards Institute (ETSI) launched a new standardization group on the so called \textit{Mobile-Edge Computing} (MEC), whose aim is to provide information technology   and cloud-computing capabilities within the Radio Access Network (RAN) in close proximity to mobile subscribers in order to offer a service environment characterized by proximity, low latency, and high rate access \cite{ETSI-MEC}.
%Bringing {\it both} radio and computational resources  closer to the MUs makes possible an efficient and scalable usage of available resources: Depending on which strategy is more convenient, from energy or latency perspectives, MUs' requests can be met on the mobile device, or on the nearby base station, or on the distant cloud.

Merging MEC with the dense deployment of (small cell) Base Stations (BSs), as foreseen in the 5G
standardization roadmap, makes possible a real proximity, ultra-low latency access to cloud functionalities \cite{Barbarossa_SPM2014}. However, in a dense deployment scenario, offloading becomes much more complicated because of  intercell interference. The goal of this paper is to propose a \emph{joint} optimization of radio and computational resources for computation offloading in a dense deployment scenario, \emph{in the presence of intercell interference}.
More specifically, the offloading problem is formulated as the minimization of the overall energy consumption, at the mobile terminals' side, under transmit power and latency constraints. %, in an intercell interference scenario.
The optimization variables are the mobile radio resources$-$the precoding (equivalently, covariance) matrices of the mobile MIMO transmitters$-$and the computational resources$-$the CPU cycles/second assigned  by the cloud to each  MU. The latency constraint is what couples computation and communication optimization variables. This problem is much more challenging than the (special) cases studied in the literature because of the presence of intercell interference, which  introduces a coupling among the precoding matrices of all  MUs,  while making the optimization problem nonconvex.
  % The offloading decision is based on the minimum energy consumption: the mobile user decides to perform computation locally or to offload it to the server depending on which strategy entails less energy consumption,  while meeting the application delay constraint.
%Nevertheless, we provide three basic important results
In this context, the main contributions of the paper are the following: i) in the single-user case,  we first establish the equivalence between the original nonconvex problem and a \emph{convex one}, and then derive the  {\it closed form} of its (global optimal) solution; ii)
in the multi-cell case, hinging on recent Successive Convex Approximation (SCA)  techniques \cite{Scutari_ICASSP14,Scutari_nonconvex}, we devise an iterative algorithm that is proved to converge to local optimal solutions of the original nonconvex problem; and  iii) we propose alternative decomposition algorithms to solve the original centralized problem in a distributed form, requiring  limited signaling among BSs and cloud; the algorithms differ for convergence speed, computational effort, communication overhead, and a-priori knowledge of system parameters, but they are all convergent under a unified set of conditions. Numerical results show that all the proposed  schemes converge quite fast to ``good'' solutions, yielding a significant energy saving with respect to disjoint optimization procedures, for applications requiring intensive computations and limited exchange of data to enable offloading.
The rest of the paper is organized as follows. In Section \ref{section:offloading} we introduce the system model;
 Section  \ref{Single_user} formulates the offloading optimization problem in the single  user case, whereas  Section \ref{section:offloading_multiple_cells} focuses on  the multi-cell scenario along with the proposed SCA algorithmic framework. The  decentralized implementation  is discussed in  Section \ref{section:decentralized}.\vspace{-0.2cm}
\section{Computation offloading}
\label{section:offloading}
Let us consider a network composed of $N_c$
cells; in each cell $n=1,\ldots, N_c$,  there is one Small Cell  enhanced Node B (SCeNB in LTE terminology) serving  $K_n$ MUs. We denote by $i_n$ the $i$-th user in the cell $n$,  and by $\mathcal I\triangleq \{i_n\,:\, i=1,\ldots , K_n,\,n=1,\ldots, N_c\}$ the set of all the users.    Each MU $i_n$ and SCeNB $n$ are equipped with $n_{T_{i_n}}$ transmit and $n_{R_n}$ receive antennas, respectively. The SCeNB's  are all  connected to a common  cloud provider, able to   serve  multiple  users concurrently.  We assume  that  MUs in the same cell transmit over orthogonal channels, whereas  users of different cells  may interfere against each other. \\ \indent In this scenario, each MU $i_n$ is willing to run an application    within a given maximum time $T_{i_n}$, while minimizing the energy consumption at the MU's side.
To offload computations to the remote cloud, the MU has to send all the needed information to the server.
Each module to be executed is characterized by: the number $w_{i_{n}}$ of  CPU cycles necessary to run the module itself; the number $b_{i_n}$ of input bits necessary to transfer the program execution from local to remote sides; and the number $b^{\texttt{o}}_{ i_n}$ of output bits encoding the result of the computation, to be sent back from remote to local sides.
\indent The MU can perform  its
computations locally or   offload them  to the cloud, depending on which strategy requires less
energy,  while satisfying the latency constraint.
 % In case of computations performed locally, at the mobile side, the latency
%is simply the ratio between the number of CPU cycles  to be executed and the number $f_{loc}$
%of CPU cycles/second that can be run locally.
In case of offloading,  the latency  incorporates
the time to transmit the input bits to the server,  the time necessary for the server to execute
the instructions, and the time to send the result back to the MU.
More specifically, the overall latency experienced by each  MU $i_n$ can be written as
\beq\label{Avg_Delay}
\Delta_{i_n}={\Delta}^{\texttt{t}}_{i_n}+\Delta^{\texttt{exe}}_{i_n}+\Delta^{\texttt{tx/rx}}_{i_n}
\eeq
where ${\Delta}^{\texttt{t}}_{i_n}$ is the time necessary for the  MU $i_n$ to transfer the input bits $b_{i_n}$  %(associated with the set of instructions to be run on the cloud)
to its SCeNB; $\Delta^{\texttt{exe}}_{i_n}$ is the time for the server to   execute $w_{i_n}$ CPU cycles; and $\Delta^{\texttt{tx/rx}}_{i_n}$ is the   time  necessary for SCeNB $n$ to send the $b_{i_n}$ bits
to  the cloud through the backhaul link plus the time necessary to send back the result (encoded in $b^{\texttt{o}}_{ i_n}$ bits)  from the server to MU $i_n$.
%We assume that   the backhaul link is a dedicated high speed connection (e.g., fiber optics) with constant latency, resulting thus in a constant   $\Delta^{\texttt{tx/rx}}_{i_n}$.
We derive next an explicit expression of ${\Delta}^{\texttt{t}}_{i_n}$ and $\Delta^{\texttt{exe}}_{i_n}$ as a function of the radio and computational resources.\\
\noindent\textbf{Radio resources}: The optimization variables at radio level are the  users' transmit covariance matrices   $\mathbf{Q}\triangleq (\mathbf{Q}_{i_n})_{i_n\in \mathcal I}$, subject to power budget constraints
\begin{equation}
\mathcal{Q}_{i_n}\triangleq\left\{
\mathbf{Q}_{i_n}\in \mathbb{C}^{n_{T_{i_n}}\times n_{T_{i_n}}}:\mathbf{Q}_{i_n}\succeq\mathbf{0},\,\,\text{tr}\left(\mathbf{Q}_{i_n}\right)\leq
P_{i_n}\right\} ,\label{eq:set_Q_i}
\end{equation}
where $P_{i_n}$ is the average transmit power of user $i_n$. We will denote by $\mathcal Q$ the joint set $\mathcal Q \triangleq \prod _{i_n\in \mathcal I} \mathcal Q_{i_n}$. \\ \indent For any given profile $\mathbf{Q}\triangleq (\mathbf{Q}_{i_n})_{i_n\in \mathcal I}$,   the maximum achievable rate of MU $i_n$ is:
\begin{equation}\label{rate}
r_{i_n}(\mathbf{Q})=\log_2 \det \left(\mathbf{I}+\mathbf{H}_{i_n n}^H{\mathbf{R}}_{n}(\mathbf{Q}_{-n})^{-1}{\mathbf{H}_{i_n n}} \mathbf{Q}_{i_n}
\right) \end{equation}
where \vspace{-0.2cm}   \begin{equation}\label{eq:MUI}\mathbf{R}_{n}(\mathbf{Q}_{-n})\triangleq \mathbf{R}_{w}+  \!\!\sum_{j_m \in \mathcal{I}, m\neq n} \! \! \! \! \mathbf{H}_{j_m n} \mathbf{Q}_{j_m} \mathbf{H}_{j_m n}^H,\vspace{-0.2cm}\end{equation} is the covariance matrix of the noise $\mathbf{R}_{w}\triangleq \sigma_w^2 \mathbf{I}$ (assumed to be diagonal w.l.o.g, otherwise one can always pre-whitening the channel matrices) plus the inter-cell interference at the SCeNB $n$ (treated as additive noise); $\mathbf{H}_{i_n n}$ is the channel matrix of the uplink  $i$ in the cell $n$, whereas $\mathbf{H}_{j_m n}$ is the cross-channel matrix between the interferer MU $j$ in the cell $m$ and the  SCeNB of cell $n$; and   $\mathbf{Q}_{-n}\triangleq ((\mathbf{Q}_{j_m})_{j=1}^{K_m})_{n\neq m=1}^{N_c}$  denotes  the tuple of the covariance matrices of all users interfering with the  SCeNB $n$.\\
 \indent Given each $r_{i_n}(\mathbf{Q})$, the time  ${\Delta}^{\texttt{t}}_{i_n}$  necessary for user $i$ in cell $n$ to transmit the input bits $b_{i_n}$
 of duration $T_{b_{i_n}}$ to its  SCeNB  can be written as\vspace{-0.1cm}
 \begin{equation}\label{travel_time}
\Delta^{\texttt{t}}_{i_n} = \Delta^{\texttt{t}}_{i_n} \left(\mathbf{Q}\right)= \dfrac{c_{i_n}}{r_{i_n}(\mathbf{Q})}\vspace{-0.2cm}
 \end{equation}
 where $c_{i_n}=b_{i_n} T_{b_{i_n}}$.
The energy consumption  due to offloading is then\vspace{-0.1cm}  \begin{equation}\label{energy}{E}_{i_n}(\mathbf{Q}_{i_n},\mathbf{Q}_{-n})= \text{tr}(\mathbf{Q}_{i_n})\cdot {\Delta^{\texttt{t}}_{i_n}\left(\mathbf{Q}\right)},\end{equation}
 which depends also on the covariance matrices $\mathbf{Q}_{-n}$ of the   users in the other cells, due to the intercell interference.\\
\noindent \textbf{Computational resources}.  The  cloud provider  is able to   serve multiple  users   concurrently. The computational resources made available  by the cloud and shared among  the users are quantified in terms of  number  of CPU cycles/second, set to $f_T$; let  $f_{i_n}\geq 0$ be the fraction of $f_T$ assigned to each user $i_n$. All the $f_{i_n}$ are thus nonnegative optimization variables to be determined,  subject to the computational budget constraint $\sum_{i_n\in \mathcal I}f_{i_n}\leq f_T$. Given  the resource assignment $f_{i_n}$, the time $\Delta^{\texttt{exe}}_{i_n}$ needed to run   $w_{i_n}$ CPU cycles of user $i_n$'s instructions remotely   %(associated with the $b_{i_n}$ transferred input bits)
is then
\begin{equation} \Delta^{\texttt{exe}}_{i_n}= \Delta^{\texttt{exe}}_{i_n}\left(f_{i_n}\right)= {w_{i_n}}/{f_{i_n}}. \label{exc_time}
\end{equation}

The expression of the overall latency $\Delta_{i_n}$ [cf. (\ref{Avg_Delay}), (\ref{travel_time}), and (\ref{exc_time})] clearly shows the interplay between radio access and computational aspects, which motivates  a \emph{joint} optimization of the radio resources, the transmit covariance matrices  $\mathbf{Q}\triangleq (\mathbf{Q}_{i_n})_{i_n\in \mathcal I}$ of the MUs, and the computational resources, the computational rate allocation $\mathbf{f}\triangleq ({f}_{i_n})_{i_n\in \mathcal I}$.

We are now ready to formulate the offloading problem rigorously. We focus first on the single-user scenario (cf. Sec. \ref{Single_user}); this will allow us to shed light on the special structure of the optimal solution. Then, we will extend the formulation to   the  multiple-cells case (cf. Sec. \ref{section:offloading_multiple_cells}).\vspace{-0.1cm}

\section{The Single-user case} \label{Single_user}\vspace{-0.1cm}
 In the single-user case, there is only one active MU having  access to the cloud. In such interference-free scenario, the maximum achievable rate on the MU and energy consumption due to offloading reduce to [cf. (\ref{rate}) and (\ref{energy})]
\beq \label{rate_SU}
r(\mathbf{Q})=\log_2 \det \left(\mathbf{I}+\mathbf{H}\mathbf{Q}\mathbf{H}^{H} {\mathbf{R}}_{w}^{-1}\right)
\eeq
and\vspace{-0.2cm}
\beq \label{energy_SU}
{E}(\mathbf{Q})=c\cdot \ds \frac{\text{tr}(\mathbf{Q})}{r(\mathbf{Q})},
\eeq
 respectively, with $c=b\cdot T_b$ (for notational simplicity, we omit the user index;  $\bQ$ denotes now  the covariance matrix of the MU).

We formulate the offloading problem as the minimization of the energy spent by the MU to run its application remotely, subject to latency and transmit power constraints, as follows:
\begin{equation}
\vspace{-0.1cm}\begin{array}{llll}
\underset{\mathbf{Q},\,  {f}}{\min}
 \quad  E(\mathbf{Q}) \\
\left.\begin{array}{clll} \mbox{\!\!s.t.} & \texttt{a)}   \ds\frac{c}{r(\mathbf{Q})}+\dfrac{w}{f}
- \tilde{T}  \leq 0 \medskip\\
& \texttt{b)}   0 \leq f \leq f_T\medskip \\
& \texttt{c)}  \mbox{tr}(\mathbf{Q})\leq P_{T}, \quad \mathbf{Q}\succeq \mathbf{0}\\
  \end{array} \right\} \triangleq \mathcal{X}_s &
\end{array} \; \label{P_MIMO_SU_ener} \tag{$\mathcal P_s$}
\end{equation}
 where   a) reflects the user latency constraint $\Delta \leq {T} $  [cf. (\ref{Avg_Delay})], with $\tilde{T}$  capturing all the constant terms, i.e., $\tilde{T}\triangleq {T}-\Delta^\texttt{tx/rx}$;
 b) imposes a limit on the cloud computational resources made available to the users; and c) is the power budget  constraint on the radio resources.

\noindent \textbf{Feasibility:} Depending on the system parameters, problem \ref{P_MIMO_SU_ener} may  be feasible or not. In the latter case, offloading is not possible and thus the  MU will perform its computations locally. It is not difficult to prove that the   following condition is  \emph{necessary} and \emph{sufficient}   for $\mathcal X_s$ to be nonempty and thus for offloading to be feasible:%$\tilde{T}>0$
\begin{equation}\label{eq:feasibility_su}
 \frac{c}{r^{\max}}+\dfrac{w}{f_T}
- \tilde{T}  \leq 0
\end{equation}
where $r^{\max}$ is the capacity of the MIMO link of the MU, i.e.,
\begin{equation}\label{eq:MIMOWF}
r^{\max}=\underset{\mathbf{Q}\succeq\mathbf{0}\,:\,\text{{tr}}(\mathbf{Q})\leq P_{T}}{\text{{argmax}}}r(\mathbf{Q}).
\end{equation}
The unique (closed-form) solution of (\ref{eq:MIMOWF}) is the well-known MIMO water-filling. Note that condition \eqref{eq:feasibility_su} has an interesting physical interpretation: offloading is feasible if and only if $\tilde{T}>0$, i.e., the delay on the wired network $\Delta^\texttt{tx/rx}$ is less than the maximum tolerable delay, and the overall latency constraint  is met (at least) when the  wireless and computational resources are  fully utilized (i.e., $r(\bQ)=r^{\max}$, and $f=f_T$). It is not difficult to check that this worst-case scenario is in fact achieved  when \eqref{eq:feasibility_su} is satisfied with equality; in such a case,  the (globally optimal) solution $(\bQ^\star, f^\star)$ to \ref{P_MIMO_SU_ener}  is trivially given by $(\bQ^\star, f^\star)=(\bQ^{\texttt{wf}}, f_T)$, where $\bQ^{\texttt{wf}}$ is the waterfilling solution to \eqref{eq:MIMOWF}. Therefore in the following we will focus w.l.o.g. on \ref{P_MIMO_SU_ener} under the tacit assumption of \emph{strict} feasibility [i.e., the inequality in \eqref{eq:feasibility_su} is tight].

\noindent \textbf{Solution Analysis:} Problem \ref{P_MIMO_SU_ener} is nonconvex due to the non-convexity of the energy function.  A major contribution of this section is to i) cast \ref{P_MIMO_SU_ener} into a convex equivalent problem, and ii) compute its global optimal solution (and thus optimal also to \ref{P_MIMO_SU_ener})  in closed form.  To do so, we introduce first some preliminary definitions.

 Let \ref{P_MIMO_SU} be the following auxiliary \emph{convex} problem
\begin{equation}
\begin{array}{llll}
\underset{\mathbf{Q},  {f}}{\min}
\quad \text{tr}(\mathbf{Q}) \\
\left.\begin{array}{llll} \mbox{\!\!s.t.} & \texttt{a)}   \ds \frac{c}{r(\mathbf{Q})}+\dfrac{w}{f}
- \tilde{T}  \leq 0 \medskip\\
& \texttt{b)}   0 \leq f \leq f_T\medskip\\
& \texttt{c)}  \mbox{tr}(\mathbf{Q})\leq P_{T}, \quad \mathbf{Q}\succeq \mathbf{0}\\
  \end{array} \right\} =\mathcal{X}_s &
\end{array} \; \label{P_MIMO_SU} \tag{$\mathcal Q_s$}
\end{equation}
which   corresponds to minimizing the transmit power of the MU under the same latency and power constraints as in \ref{P_MIMO_SU_ener}. Also, let   $\mathbf{H}^{H} {\mathbf{R}}_{w}^{-1}\mathbf{H}=\mathbf{U}\mathbf{D}\mathbf{U}^H$ be the (reduced) eigenvalue decomposition of $\mathbf{H}^{H} {\mathbf{R}}_{w}^{-1}\mathbf{H}$, with $r\triangleq \text{rank}(\mathbf{H}^{H} {\mathbf{R}}_{w}^{-1}\mathbf{H})=\text{rank}(\bH)$, where $\mathbf{U}\in \mathbb{C}^{n_T\times r}$  is the (semi-)unitary matrix whose columns are the eigenvectors associated with the $r$ positive eigenvalues of
  $\mathbf{H}^{H} {\mathbf{R}}_{w}^{-1}\mathbf{H}$, and $\mathbb{R}_{++}^{r\times r}\ni\bD\triangleq \text{diag}\{(d_i)_{i=1}^{r}\}$ is the diagonal matrix, whose diagonal entries are the eigenvalues  arranged in decreasing order. We are now ready to establish the connection between \ref{P_MIMO_SU_ener}  and \ref{P_MIMO_SU}.

\begin{theorem}\label{thm:MIMO_SU_SC_energy}\it{  Given
 problems \ref{P_MIMO_SU_ener}  and \ref{P_MIMO_SU}  under strict feasibility, the following hold.

\noindent \emph{(a)}  \ref{P_MIMO_SU_ener}  and \ref{P_MIMO_SU} are equivalent;\smallskip

\noindent \emph{(b)} \ref{P_MIMO_SU} (and \ref{P_MIMO_SU_ener}) has a unique solution $(\bQ^{\star},f^\star)$, given by\vspace{-0.2cm}
\begin{equation}
f^\star=f_T,\quad \text{and}\quad  \mathbf{Q}^{\star}=\mathbf{U}\left(\alpha \mathbf{I} -\mathbf{D}^{-1}\right)^{+}\mathbf{U}^H, \label{MIM0_SU}\end{equation}
where    $\alpha>0$ must be chosen so  that the latency constraint (a) in $\mathcal X_s$   is satisfied with equality at  $(\bQ^{\star},f^\star)$, and $(\bx)^+\triangleq \max(\mathbf{0}, \bx)$ (intended component-wise).

The water-level $\alpha>0$ can be efficiently computed using the  hypothesis-testing-based algorithm described in Algorithm \ref{algorithm:Alg_SU_SC_MIMO}.}\end{theorem}
\begin{proof} See Appendix \ref{A:proof Th1}.\end{proof}\vspace{-0.3cm}

\begin{algorithm}[H]

\textbf{Data:} $(d_i)_{i=1}^r>\mathbf{0}$ (arranged in decreasing order),   $r=\mbox{rank}(\mathbf{H}^{H} {\mathbf{R}}_{w}^{-1} \mathbf{H})$, and $L\triangleq \tilde{T}-w/f_T>0$;

(\texttt{S.0}): {Set} $r_e=r$;

(\texttt{S.1}): {Repeat}\vspace{-0.2cm}

\quad \quad \quad   \,\,(a): {Set} $\alpha= 2^{\ds \frac{c}{r_e L}-\ds \frac{1}{r_e}\ds\sum_{i=1}^{r_e} \log_2(d_i)};$\bigskip

\quad \quad \quad       \,\,(b): { If }  $p_i\triangleq (\alpha - 1/d_i)\geq 0$,  $\forall i=1,\ldots,r_e$, \smallskip

\quad \quad \quad \hspace{1.2cm} {and} \,$\sum_{i=1}^{r_e}p_i\leq P_T$,

\quad \quad \quad    \quad \quad  \, then \texttt{STOP};\smallskip

 \quad \quad \quad    \quad \quad   \,\,\,{else}\, $r_e=r_e-1$; \smallskip

  \quad \quad \quad   {until}   $r_e\geq 1$.
       \caption{Efficient computation of $\alpha$ in (\ref{MIM0_SU}) }
 \label{algorithm:Alg_SU_SC_MIMO}
\end{algorithm}\vspace{-0.2cm}

Theorem 1 is the formal proof that, in the single-user case, the latency constraint has to be met with equality and then the offloading strategy minimizing energy consumption coincides with the one minimizing the transmit power. Note also that  $\bQ^\star$  has a water-filling-like structure: the optimal transmit ``directions'' are aligned with the eigenvectors $\bU$ of the equivalent channel $\mathbf{H}^{H} {\mathbf{R}}_{w}^{-1}\mathbf{H}$. However, differently from the classical waterfilling solution $\bQ^{\text{wf}}$  [cf. (\ref{eq:MIMOWF})], the waterlevel $\alpha$  is now computed to meet the latency constraints   with equality. This means  that a transmit strategy using the full power $P_T$ (like $\bQ^{\text{wf}}$) is no longer  optimal. %, \textcolor{red}{ when    $P_T> P_{\max}=\sum_{i=1}^{r}(\alpha(r)-1/d_i)^{+}$, with $\alpha(r)$ being the solution of (a) in Step 1 of Algorithm \ref{algorithm:Alg_SU_SC_MIMO} corresponding to $r_e=r$}.
The only case in which  $\bQ^\star\equiv \bQ^{\text{wf}}$ is the case where the feasibility condition (\ref{eq:feasibility_su}) is satisfied with equality. Note also that the water-level $\alpha$ depends now on \emph{both} communication and computational parameters (the maximum tolerable delay,  size of the program state, CPU cycle budget, etc.).\vspace{-0.2cm}

\section{Computation offloading over multiple-cells}
\label{section:offloading_multiple_cells}
 In this section we consider   the more general multi-cell scenario described in Sec.\ref{section:offloading}.
The overall energy   spent by the MUs to remotely run their applications is now given by\vspace{-0.2cm}

 \begin{equation}{E}(\mathbf{Q})\triangleq \displaystyle{\sum_{i_n \in \mathcal{I}}} E_{i_n}(\mathbf{Q}),\label{MUs_energy}\vspace{-0.2cm}
\end{equation}
with $E_{i_n}(\mathbf{Q})$ defined in (\ref{energy}). If some fairness has to be guaranteed among the MUs, other objective functions of the MUs' energies  $E_{i_n}(\mathbf{Q})$ can be used, including the weighted sum, the (weighted) geometric mean, etc.. As a case-study, in the following, we  will focus on the minimization of the sum-energy ${E}(\mathbf{Q})$, but the proposed algorithmic framework can be readily applied to the alternative aforementioned  functions.

Each MU  $i_n$ is subject to the power budget constraint  \eqref{eq:set_Q_i} and, in case of  offloading, to an overall latency given by \vspace{-0.1cm}
\begin{equation}g_{i_n}(\mathbf{Q},f_{i_n})\triangleq \dfrac{c_{i_n}}{r_{i_n}(\mathbf{Q})} + \dfrac{w_{i_n}}{f_{i_n}}- \tilde{T}_{i_n}\leq 0.\label{MUs_latency}\vspace{-0.1cm}
\end{equation}

The offloading problem in the multi-cell scenario is then formulated as follows:
  \begin{equation}
 \begin{array}{llll}
\underset{\mathbf{Q},  \mathbf{f}}{\min}
\quad  {E}(\mathbf{Q})\\
\left.\begin{array}{llll} \mbox{\!\!s.t.} & \texttt{a)}\,
g_{i_n}(\mathbf{Q},f_{i_n})\leq 0,\,\,\forall {i_n\in \mathcal I},\medskip \\
& \texttt{b)}\,  \displaystyle{\sum_{i_n\in \mathcal I}}f_{i_n}\leq f_T,\quad f_{i_n}\geq 0,\quad \forall {i_n\in \mathcal I},\medskip\\
& \texttt{c)}\,  \mathbf{Q}_{i_n}\in \mathcal{Q}_{i_n},\quad \forall {i_n\in \mathcal I},\end{array}\right\} \triangleq \mathcal{X} &
\end{array}\label{prob_ener}\tag{$\mathcal P$}
\end{equation}
where   a) represent the users' latency constraints $\Delta_{i_n}\leq {T}_{i_n}$ with  $\tilde{T}_{i_n}\triangleq {T}_{i_n}-\Delta_{i_n}^\texttt{tx/rx}$;
and    the constraint in b) is due to the limited cloud computational resources to be allocated among the MUs.

\noindent \textbf{Feasibility}:  The  following conditions are sufficient  for $\mathcal X$ to be nonempty and thus for offloading to be feasible: $\tilde{T}_{i_n}>0$ for all ${i_n\in \mathcal I}$, and there exists a $\bar{\mathbf{Q}}\triangleq (\bar{\mathbf{Q}}_{i_n})_{i_n\in \mathcal I}\in \mathcal{Q}$ such that
\begin{equation}
\tilde{T}_{i_n}>  \ds \frac{c_{i_n}}{r_{i_n}(\bar{\mathbf{Q}})},\,\forall i_n\in \mathcal I,\quad \text{and}\quad    \ds  \sum_{i_n\in \mathcal I} \ds \frac{w_{i_n}}{\tilde{T}_{i_n}-\ds \frac{c_{i_n}}{{r_{i_n}}(\bar{\mathbf{Q}})}}\leq f_T.\label{feasibility_multicell}
\end{equation}

 Problem $\mathcal{P}$ is nonconvex, due to the nonconvexity of the objective function and the constraints a).  In what follows we  exploit the structure of $\mathcal{P}$ and,  building on some recent Successive Convex Approximation (SCA) techniques proposed in \cite{Scutari_ICASSP14,Scutari_nonconvex}, we develop a fairly general class of  efficient  approximation algorithms, all  converging to a local optimal solution of $\mathcal{P}$. The numerical results will show that the proposed algorithms converge in a few iterations to ``good'' locally optimal solutions of  $\mathcal{P}$ (that turn out to be quite insensitive to the initialization). The main algorithmic framework, along with its convergence properties, is introduced in Sec. \ref{section:Algorithmic_design}; alternative distributed implementations are studied in Sec. \ref{section:decentralized}.\vspace{-0.3cm}

 \subsection{Algorithmic design}
 \label{section:Algorithmic_design}
   To solve the non-convex problem $\mathcal{P}$ efficiently, we develop a SCA-based method where $\mathcal{P}$ is replaced by  a sequence of \emph{strongly convex}   problems. At the basis of the proposed technique, there is a suitable \emph{convex}  approximation of the nonconvex objective function $E(\mathbf{Q})$ and the constraints $g_{i_n}(\mathbf{Q},f_{i_n})$ around the iterates of the algorithm, which are preliminarily discussed next.
\subsubsection{Approximant of $E(\mathbf{Q})$} Let $\mathbf{Z}\triangleq (\bQ, \mathbf{f})$ and  $\mathbf{Z}^{\nu}\triangleq (\mathbf{Q}^{\nu}, \mathbf{f}^{\nu})$, with $\mathbf{f}\triangleq (f_{i_n})_{i_n \in \mathcal I}$ and  $\mathbf{f}^{\nu}\triangleq (f_{i_n}^{\nu})_{i_n \in \mathcal I}$. Let $\mathcal E \supseteq \mathcal X$ be any closed convex set containing $\mathcal X$ such that $E(\bQ)$ is well-defined on it. Note that such a set exits. For instance, noting that  at  every (feasible) $(\mathbf{Q}, \mathbf{f})\in \mathcal X$, it must be $r_{i_n}(\mathbf{Q})>0$,  $f_{i_n}>0$, for all $i$ and $n$. Hence, condition $g_{i_n}(\mathbf{Q},f_{i_n})\leq 0$ in \ref{prob_ener} can be equivalently rewritten as
\[
r_{i_n}(\mathbf{Q})  \geq  \alpha_{i_n}(f_{i_n})\triangleq \frac{c_{i_n}\cdot f_{i_n}}{f_{i_n}\cdot \tilde{T}_{i_n}-w_{i_n}}>0,
\]
so that one can choose  $\mathcal E\triangleq \{ (\bQ,\mathbf{f}) : \text{b)},\, \text{c)}\, \text{hold},\, \, r_{i_n}(\bQ_{i_n},\bQ_{-i_n}=\mathbf{0})\geq  \alpha_{i_n}(f_{i_n}), \,\forall i_n\in \mathcal I \}$.

 Following \cite{Scutari_ICASSP14,Scutari_nonconvex}, our goal is to
build, at each iteration $\nu$,  an approximant, say $\tilde{E}(\mathbf{Z};\mathbf{Z}^\nu)$,  of the   nonconvex (nonseparable)  $E(\mathbf{Q})$   around the current (feasible) iterate $\mathbf{Z}^\nu \in \mathcal{X}$ that enjoys the following key properties: %\textcolor{red}{(see. Appendix B for more details)}:
\begin{description}
\item[P1:]
  $\tilde{E}(\bullet;\mathbf{Z}^\nu)$  is uniformly  \emph{strongly convex} on $\mathcal E \times \mathbb{R}_+^{|\mathcal I|}$;\smallskip

\item[P2:] $\nabla_{\mathbf Q^\ast} \tilde{E}(\bZ^{\nu};\bZ^{\nu})=\nabla_{\bQ^\ast}{E}(\bQ^{\nu})$,  $\forall\bZ^{\nu}\in \mathcal X$; \smallskip

\item[P3:] $\nabla_{\bZ^\ast} \tilde{E}(\bullet;\bullet)$ is Lipschitz continuous on $\mathcal E \times \mathbb{R}_+^{|\mathcal I|} \times \mathcal X$; \end{description}
where $\nabla_{\bZ^\ast} \tilde{E}(\bullet;\bullet)$ denotes the conjugate gradient of $\tilde{E}$ with respect to $\bZ$.  Conditions P1-P2 just guarantee that the candidate approximation $\tilde{E}(\bullet;\mathbf{Z}^\nu)$ is strongly convex while preserving the  same first order behaviour of  $E(\mathbf{Q})$ at any iterate  $\mathbf{Q}^\nu $; P3 is a standard continuity requirement.

We build next a  $\tilde{E}(\mathbf{Z};\mathbf{Z}^\nu)$ satisfying P1-P3. Observe that i) for any given $\mathbf{Q}_{-n}=\mathbf{Q}_{-n}^\nu$, each term   $E_{i_n}(\mathbf{Q}_{i_n},\mathbf{Q}_{-n}^\nu)=\text{tr}(\bQ_{i_n})\cdot \Delta_{i_n}^{\texttt{t}}(\mathbf{Q}_{i_n},\mathbf{Q}_{-n}^\nu)$ of the sum in $E(\mathbf{Q})$   [cf. (\ref{MUs_energy})]
 % \begin{equation}\label{E_in}E_{i_n}(\mathbf{Q}_{i_n},\mathbf{Q}_{-n}^\nu)= \text{tr}(\mathbf{Q}_{i_n})\cdot \Delta_{i_n}^\texttt{bkh} + \dfrac{b_{i_n}\cdot \text{tr}(\mathbf{Q}_{i_n})}{r_{i_n}(\mathbf{Q}_{i_n},\mathbf{Q}_{-n}^\nu)}\end{equation}
 is the product of two convex functions in  $\mathbf{Q}_{i_n}$ [cf. (\ref{energy})], namely: $\text{tr}(\bQ_{i_n})$ and $\Delta_{i_n}^{\texttt{t}}(\mathbf{Q}_{i_n},\mathbf{Q}_{-n}^\nu)$;  and ii)  the other terms of the sum$-\sum_{j_m \in \mathcal{I}, m\neq n} E_{j_m}(\mathbf{Q}_{i_n},\mathbf{Q}_{-i_n,j_m}^{\nu})$ with $\mathbf{Q}_{-i_n,j_m}^{\nu} \triangleq (\mathbf{Q}_{j_{m}}^{\nu},(\mathbf{Q}_{l_q}^{\nu})_{\forall l,q \neq m, l_q \neq i_n})-$are not convex  in $\mathbf{Q}_{i_n}$. Exploiting such a structure, a convex approximation of  $E(\mathbf{Q})$  can be obtained for each MU $i_n$  by   convexifying the  term $\text{tr}(\bQ_{i_n})\cdot \Delta_{i_n}^{\texttt{t}}(\mathbf{Q}_{i_n},\mathbf{Q}_{-n}^\nu)$  and    linearizing  the nonconvex part $ \sum_{j_m \in \mathcal{I}, m\neq n} E_{j_m}(\mathbf{Q}_{i_n};\mathbf{Q}_{-i_n, j_m}^{\nu})$. More formally, denoting  $\mathbf{Z}_{i_n}\triangleq (\bQ_{i_n}, f_{i_n})$, for each $i_n$, let us introduce the ``approximation'' function  $\tilde{E}_{i_n}(\mathbf{Z}_{i_n};\mathbf{Q}^\nu)$:
\begin{equation}\label{E_tilde_in}
\hspace{-0.1cm}\begin{array}{lll} \tilde{E}_{i_n}(\mathbf{Z}_{i_n};\mathbf{Z}^\nu)\!\!\!\!\!&\triangleq & \!\!\!\!   \dfrac{c_{i_n}\cdot \text{tr}(\mathbf{Q}_{i_n})}{r_{i_n}(\mathbf{Q}_{i_n}^\nu,\mathbf{Q}_{-n}^\nu)} +\dfrac{c_{i_n}\cdot \text{tr}(\mathbf{Q}_{i_n}^\nu)}{r_{i_n}(\mathbf{Q}_{i_n},\mathbf{Q}_{-n}^\nu)}\medskip \\
\!\!&&\!\!\!\! +\displaystyle{\sum_{j_m \in \mathcal{I}, m\neq n }} \left\langle\nabla_{\mathbf{Q}_{i_n}^\ast}E_{j_m}(\mathbf{Q}^\nu),  \mathbf{Q}_{i_n}-\mathbf{Q}_{i_n}^\nu\right\rangle \medskip\\
\!\!&& \!\!\!\! + {\tau_{i_n}} \, \|\mathbf{Q}_{i_n}-\mathbf{Q}_{i_n}^\nu \|^2 +\dfrac{{c_{f_{i_n}}}}{2}  \, (f_{i_n}-f_{i_n}^\nu)^2
\end{array}
\end{equation}
 where: the first two terms on the right-hand side are the aforementioned convexification of $\text{tr}(\bQ_{i_n})\cdot \Delta_{i_n}^{\texttt{t}}(\bQ_{i_n},\bQ_{-i_n}^\nu)$; the third term comes from the linearization of  $\sum_{j_m \in \mathcal{I}, m\neq n} E_{j_m}(\mathbf{Q}_{i_n};\mathbf{Q}_{-i_n, j_m}^{\nu})$, with  $\left\langle \mathbf{A}, \mathbf{B}\right\rangle\triangleq  \text{Re}\{ \text{tr} (\mathbf{A}^H\mathbf{B})\}$ and $\nabla_{\mathbf{Q}_{i_n}^\ast}E_{j_m}(\mathbf{Q}^\nu)$ denoting the conjugate gradient of $E_{j_m}(\mathbf{Q})$ with respect to $\mathbf{Q}_{i_n}$ evaluated at $\mathbf{Q}^\nu$, and given by
 \begin{equation}
\begin{array}{lll}\nabla_{\mathbf{Q}_{i_n}^\ast}E_{j_m}(\mathbf{Q}^\nu)=
\dfrac{\text{tr}(\mathbf{Q}^\nu_{j_m})\Delta^{t}_{j_m}(\mathbf{Q}^\nu)}{\log(2) r_{j_m}(\mathbf{Q}^\nu)} \cdot  \left[ \mathbf{H}^{ H}_{i_n m} \left({\mathbf{R}}_{m}(\mathbf{Q}_{-m}^\nu)^{-1} \right. \right. \\ \quad \quad \quad \left. \left.  -(\mathbf{R}_{m}(\mathbf{Q}_{-m}^\nu)+\mathbf{H}_{j_m m} \mathbf{Q}_{j_m}^\nu \mathbf{H}^{H}_{j_m m})^{-1}  \right)\mathbf{H}_{i_n m} \right] ;\end{array}
\end{equation}
the fourth term in \eqref{E_tilde_in} is a quadratic regularization term added to make  $\tilde{E}_{i_n}(\bullet;\mathbf{Z}^\nu)$ uniformly strongly convex on $\mathcal E \times \mathbb{R}_{+}$. \\
\noindent Based on each $\tilde{E}_{i_n}(\mathbf{Z}_{i_n};\mathbf{Z}^\nu)$, we can now define the candidate sum-energy approximation   $\tilde{E}(\mathbf{Z};\mathbf{Z}^\nu)$ as: given $\mathbf{Z}^\nu \in \mathcal X$,
\begin{equation}\label{E_tilde}\vspace{-0.3cm}
\tilde{E}(\mathbf{Z};\mathbf{Z}^\nu)\triangleq \displaystyle{\sum_{i_n \in \mathcal{I}}} \tilde{E}_{i_n}(\mathbf{Z}_{i_n};\mathbf{Z}^\nu).
\end{equation}

It is not difficult to check that   $\tilde{E}(\mathbf{Z};\mathbf{Z}^{\nu})$ satisfies P1-P3; in particular it is strongly convex on $\mathcal E \times \mathbb{R}_+^{|\mathcal I|}$  with constant $c_{\tilde{E}}\geq \min_{i_n \in \mathcal I} (\min(\tau_{i_n},c_{f_{i_n}}))\!\!>\!0$. Note that $\tilde{E}(\mathbf{Z};\mathbf{Z}^{\nu})$ is also
 separable in the users variables $\mathbf{Z}_{i_n}$, which is instrumental to obtain distributed algorithms across the   SCeNBs, see Sec.  \ref{section:decentralized}.

\subsubsection{Inner convexification of the constraints $g_{i_n}(\mathbf{Q},f_{i_n})$}
We aim at introducing an inner convex approximation, say  $\tilde{g}_{i_n}(\mathbf{Q},f_{i_n};\mathbf{Z}^\nu)$,  of the constraints $g_{i_n}(\mathbf{Q},f_{i_n})$ around $\mathbf{Z}^{\nu}\in \mathcal{X}$,   satisfying the following key properties (the proof is omitted for lack of space and reported in  Appendix B in the supporting material) \cite{Scutari_ICASSP14,Scutari_nonconvex}:
\begin{description}
   \item [C1:]$\tilde{g}_{i_n} (\bullet; \mathbf{Z}^{\nu})$ is uniformly convex on $\mathcal E \times \mathbb{R}_+$;\smallskip
   \item [C2:] $\nabla_{\mathbf{Z}^{\ast}}\tilde{g}_{i_n} (\bQ^\nu,f_{i_n}^\nu; \mathbf{Z}^\nu)=\nabla_{\mathbf{Z}^{\ast}}{g}_{i_n} (\bQ^\nu,f_{i_n}^\nu)$,
   $\forall \mathbf{Z}^\nu \in \mathcal{X}$;\smallskip
    \item [C3:] $\nabla_{\mathbf{Z}^{^\ast}}\tilde{g}_{i_n} (\bullet;\bullet)$  is continuous on
   $\mathcal E \times \mathbb{R}_+\times  \mathcal{X}$;\smallskip
    \item [C4:] $ \tilde{g}_{i_n} (\mathbf{Q},{f}_{i_n}; \mathbf{Z}^\nu) \geq g_{i_n} (\mathbf{Q},{f}_{i_n})$,
  $\forall (\mathbf{Q},{f}_{i_n})\in \mathcal E \times \mathbb{R}_+$ and  $\forall\mathbf{Z}^\nu \in \mathcal{X}$;\smallskip

  \item [C5:]  $ \tilde{g}_{i_n} (\mathbf{Q}^\nu,f_{i_n}^\nu; \mathbf{Z}^\nu) = g_{i_n} (\mathbf{Q}^\nu,f_{i_n}^\nu)$,
   $\forall \mathbf{Z}^\nu \in \mathcal{X}$;\smallskip

     \item [C6:] $\tilde{g}_{i_n} (\bullet; \bullet)$  is Lipschitz continuous on
   $\mathcal E \times \mathbb{R}_+\times \mathcal{X}$.
 \end{description}
%where $\nabla_{\mathbf{Z}^\ast} \tilde{g}_{i_n} (\bullet; \bullet)$ denotes the conjugate gradient of
%$\tilde{g}_{i_n} (\bullet; \bullet)$
 %with respect
%to $\mathbf{Z}$.

Conditions C1-C3 are the counterparts of P1-P3 on $\tilde{g}_{i_n}$; the extra condition C4-C5 guarantee that $\tilde{g}_{i_n}$ is an inner approximation of ${g}_{i_n}$, implying that any $(\bQ,f_{i_n})$ satisfying $\tilde{g}_{i_n}(\bQ, f_{i_n};\bZ^\nu)\leq 0$ is feasible also for the original nonconvex problem $\mathcal P$.

 To build a   $\tilde{g}_{i_n}$ satisfying C1-C6, let us exploit first the  concave-convex structure of the rate functions $r_{i_n}(\mathbf{Q})$ [cf. (\ref{rate})]:
\begin{equation}\label{DC_rate}
r_{i_n}(\mathbf{Q})=r_{i_n}^{\,\text{+}}(\mathbf{Q})+r_{n}^{\,\text{-}}(\mathbf{Q}_{-n}),
\end{equation}
where
\begin{equation}\label{r_lus}\begin{array}{lll}
r_{i_n}^{\,\text{+}}(\mathbf{Q})\triangleq  \log_2 \det \left(\mathbf{R}_{n}(\mathbf{Q}_{-n})+\mathbf{H}_{i_n n}\mathbf{Q}_{i_n}\mathbf{H}_{i_n n}^H  \right)\medskip \\
r_{n}^{\,\text{-}}(\mathbf{Q}_{-n})\triangleq  -
\log_2 \det \left(\mathbf{R}_{n}(\mathbf{Q}_{-n})\right)\smallskip
\end{array}
\end{equation}
with    $\mathbf{R}_{n}(\mathbf{Q}_{-n})$  defined in (\ref{eq:MUI}). Note that  $r_{i_n}^{\,\text{+}}(\bullet)$ and $r_{n}^{\,\text{-}}(\bullet)$ are concave on   $\mathcal{Q}$ and convex on  $\mathcal{Q}_{-n}\triangleq \prod_{m\neq n} \mathcal{Q}_m$, respectively.  Using \eqref{DC_rate}, and observing that at any  (feasible) $(\mathbf{Q}, \mathbf{f})\in \mathcal X$, it must be  $r_{i_n}(\mathbf{Q})>0$ and $f_{i_n}>0$ for all $i$ and $n$, the constraints $g_{i_n}(\mathbf{Q},f_{i_n})\leq 0$ in  \ref{prob_ener} can be equivalently rewritten as
\begin{equation} \label{g_DC}
g_{i_n}(\mathbf{Q},f_{i_n}) = - r_{i_n}^{\,\text{+}}(\mathbf{Q}) - r_{n}^{\,\text{-}}(\mathbf{Q}_{-n}) + \frac{c_{i_n}\cdot f_{i_n}}{f_{i_n}\cdot \tilde{T}_{i_n}-w_{i_n}}\leq 0,
\end{equation}
where with a slight abuse of notation we used the same symbol $g_{i_n}(\mathbf{Q},f_{i_n})$ to denote the constraint in the equivalent form.

  The desired inner convex approximation  $\tilde{g}_{i_n}(\mathbf{Q},f_{i_n};\mathbf{Z}^\nu)$ is obtained from $g_{i_n}(\mathbf{Q},f_{i_n})$ by retaining the convex part in (\ref{g_DC}) and linearizing the concave term $- r_{n}^{\,\text{-}}(\mathbf{Q}_{-n})$, resulting in:
\begin{equation}\label{g_tilde}
\begin{array}{lll}
\!\!\!\tilde{g}_{i_n}(\mathbf{Q},f_{i_n};\mathbf{Z}^\nu)\!\!\!& \!\!\!\triangleq
- r_{i_n}^{\,\text{+}}(\mathbf{Q}) + \dfrac{c_{i_n}\cdot f_{i_n}}{f_{i_n}\cdot \tilde{T}_{i_n}-w_{i_n}} \medskip\\
\!\!\!& \hspace{0.2cm} - r_{n}^{\,\text{-}}(\mathbf{Q}_{-n}^\nu)\! - \!\!\!\!\displaystyle{\sum_{j_m \in \mathcal{I}}} \!\! \left\langle \boldsymbol{\Pi}^{\,\text{-}}_{j_m,n}(\mathbf{Q}^{\nu}),\! \mathbf{Q}_{j_m}-\mathbf{Q}_{j_m}^{\nu} \right\rangle
\end{array}
\end{equation}
where each $\boldsymbol{\Pi}^{\,\text{-}}_{j_m,n}(\mathbf{Q}^{\nu})$ is defined as
\beq
   \boldsymbol{\Pi}_{j_{m},n}^{\,\text{-}}(\mathbf{Q}^{\nu})\triangleq\left\{ \begin{split}&\nabla_{\mathbf{Q}_{j_{m}}^{\ast}}r_{n}^{\,\text{-}}(\mathbf{Q}_{-n}^{\nu}), & \text{if}\, m\neq n;\qquad\quad\qquad\\
&\mathbf{0}, & \text{otherwise};\,\,\qquad\qquad\,\,
\end{split}
\right.
\eeq
and $\nabla_{\mathbf{Q}_{j_{m}}^{\ast}}r_{n}^{\,\text{-}}(\mathbf{Q}_{-n}^{\nu})=-\mathbf{H}^{H}_{j_m n}\mathbf{R}_{n}(\mathbf{Q}_{-n}^\nu)^{-1}\mathbf{H}_{j_m n}$.\smallskip
\subsubsection{Inner SCA algorithm: centralized implementation}
\label{section:centralized}
  We are now ready to introduce the proposed inner convex approximation of the nonconvex problem $\mathcal{P}$, which consists in replacing the nonconvex objective function $E(\mathbf{Q})$ and constraints $g_{i_n}(\mathbf{Q},f_{i_n})\leq 0$ in $\mathcal{P}$ with  the approximations $\tilde{E}(\mathbf{Z};\mathbf{Z}^\nu)$ and  $\tilde{g}_{i_n}(\mathbf{Q},f_{i_n};\mathbf{Z}^\nu)\leq 0$, respectively. More formally, given the feasible point $\mathbf{Z}^\nu$,   we have
\begin{equation}
\widehat{\mathbf{Z}}(\mathbf{Z}^\nu)\triangleq \begin{array}[t]{clll}
 \underset{\mathbf{Q},  \mathbf{f}}{ \text{argmin}}
\!\!& \! \,\tilde{E}(\mathbf{Q};\mathbf{Q}^\nu)\vspace{-1.7cm}\\
\mbox{ s.t.} \!\!& \!\!   \begin{array}{llll} {}\vspace{1.3cm}\\ \texttt{a)}\,\,
\tilde{g}_{i_n}(\mathbf{Q},f_{i_n};\mathbf{Z}^\nu)\leq 0,\,\,\forall  i_n\in \mathcal I,\medskip \\\texttt{b)}\,\, \displaystyle{\sum_{i_n\in \mathcal I}}f_{i_n}\leq f_T,\quad f_{i_n}\geq 0,\,\, \forall  i_n\in \mathcal I,\medskip\\\texttt{c)}\,\, \mathbf{Q}_{i_n}\in \mathcal{Q}_{i_n},\quad \forall i_n\in \mathcal I,\end{array}
\end{array} \label{prob_cvx}\tag{$\mathcal{P}^\nu$}\vspace{-0.1cm}
\end{equation}
where we denoted by $\widehat{\mathbf{Z}}(\mathbf{Z}^\nu)\triangleq (\widehat{\mathbf{Q}}(\mathbf{Z}^\nu),\widehat{\mathbf{f}}(\mathbf{Z}^\nu))$ the unique solution of the strongly convex optimization problem.

 The proposed solution consists in solving the sequence of problems \ref{prob_cvx}, starting from a feasible $\mathbf{Z}^0\triangleq (\mathbf{Q}^0, \mathbf{f}^0)$. The formal description of the method is given in Algorithm   \ref{alg:Alg_centr}, which is proved to converge to local optimal solutions of the original nonconvex problem $\mathcal{P}$  in Theorem \ref{thm:centralized}.
 Note that in Step 3 of the algorithm we include a memory in the update of the iterate $\mathbf{Z}^{\nu}\triangleq (\mathbf{Q}^\nu, \mathbf{f}^\nu)$. A practical termination criterion in Step
1 is   $|E(\mathbf{Q}^{\nu+1})-E(\mathbf{Q}^{\nu})|\leq \delta$, where $\delta>0$
is the prescribed accuracy.

     \begin{algorithm}[H]
\textbf{Initial  data:} $\mathbf{Z}^0\triangleq(\mathbf{Q}^{0},\mathbf{f}^0)\in \mathcal{X}$; $\{\gamma^{\nu}\}_\nu \in (0,1]$; %$c_\tau>0$; $c_f>0$. Set $\nu=0$.

(\texttt{S.1}): If $\mathbf{Z}^{\nu}$ satisfies a suitable termination criterion, \texttt{STOP}

(\texttt{S.2}): Compute    $\hat{\mathbf{Z}}(\mathbf{Z}^{\nu})\triangleq (\hat{\mathbf{Q}}(\mathbf{Z}^{\nu}),\hat{\mathbf{f}}(\mathbf{Z}^{\nu}))$ [cf. $\mathcal{P}^{\nu}$];

(\texttt{S.3}):  Set $\mathbf{Z}^{\nu+1}=\mathbf{Z}^{\nu}+\gamma^{\nu}\left(\hat{\mathbf{Z}}(\mathbf{Z}^{\nu})-\mathbf{Z}^{\nu}\right)$;

(\texttt{S.4}):   $\nu \leftarrow \nu+1$  and go  to (\texttt{S.1}).

\caption{\textbf{:}  Inner SCA  Algorithm for  $\mathcal{P}$  \label{alg:Alg_centr}}

\end{algorithm}
\begin{theorem} \label{thm:centralized}{\it
Given the nonconvex problem $\mathcal{P}$,   choose $c_{\tilde{E}}>0$ and $\{\gamma^{\nu}\}_\nu$ such that
  \begin{equation}
(0,1]\ni\gamma^{\nu}\rightarrow 0, \,\forall \nu\geq 0,\quad \mbox{and}\quad\sum_{\nu}\gamma^{\nu}=+\infty.\label{eq:diminishing_step_size}
 \end{equation}
 Then every  limit point of $\{\mathbf{Z}^{\nu}\}$ (at least one of such points exists) is a stationary solution
of $\mathcal{P}$. Furthermore, none of such points is a local maximum of the energy function  $E$.}
\end{theorem}
\begin{proof}
The proof is omitted for lack of space
and reported in Appendix B of the supporting material.
\end{proof}

Theorem
\ref{thm:centralized} offers some flexibility in the choice of the free parameters  $c_{\tilde{E}}$ and $\{\gamma^{\nu}\}_\nu$  while
guaranteeing convergence of Algorithm  \ref{alg:Alg_centr}. For instance, $c_{\tilde{E}}$ is positive if all $\tau_{i_n}$ and $c_{f_{i_n}}$ are positive (but arbitrary); in the case of full-column rank matrices $\mathbf{H}_{i_n n}$, one can also set $\tau_{i_n}=0$ (still resulting in  $c_{\tilde{E}}>0$). %Any $c_f>0$ is instead admissible.
Many choices are possible for the step-size  $\gamma^{\nu}$; a practical rule satisfying (\ref{eq:diminishing_step_size})  that we found effective in our experiments is  \cite{scutari_facchinei_et_al_tsp13}: \begin{equation}\gamma^{\nu+1}=\gamma^{\nu}(1-\alpha \gamma^{\nu}), \quad  \gamma^{0}\in (0,1],
\end{equation} with $\alpha\in \left(0,1/{\gamma^{0}}\right)$.\smallskip

\noindent \emph{On the implementation of Algorithm 2:}
 Since the base stations are connected to the cloud throughout high speed wired links, a good candidate place to run Algorithm 2 is the cloud itself: The cloud collects first  all system parameters needed to run the algorithm from the SCeNBs (MUs' channel state information, maximum tolerable latency,  etc.); then, if the feasibility conditions \eqref{feasibility_multicell} are satisfied, the cloud solves the strongly convex problems $\mathcal P^\nu$ (using any standard nonlinear programming solver), and sends the solutions  $\bQ_{n}$ back to the corresponding SCeNBs; finally, each SCeNB communicates the optimal transmit parameters to the MUs it is serving.\smallskip

\noindent \emph{Related works:}  Algorithm 2 hinges on the idea of  successive convex programming, which aims at computing stationary solutions of some classes of nonconvex problems by solving a sequence of convexified subproblems. Some relevant instances of this method  that have attracted significant interest in recent years are: i) the basic DCA (Difference-of-Convex Algorithm) \cite{ThiTao05,AlvaradoScutariPange14}; ii) the M(ajorization)-M(inimization) algorithm \cite{HunterLange04,Sriperumbudur09}; iii)  alternating/successive minimization methods \cite{Boyd13,razaviyayn2013unified,scutari_facchinei_et_al_tsp14}; and  iv)  partial linearization methods \cite{patriksson1993partial,BolteShohamTeboulle13,scutari_facchinei_et_al_tsp13}.
The aforementioned methods  identify  classes of ``favorable'' nonconvex  functions, for which a suitable convex approximation can be obtained and  convergence of the associated sequential convex programming  method can be proved.  However,  the sum-energy function $E(\mathbf{Q})$ in (\ref{MUs_energy}) and the resulting nonconvex optimization problem  $\cal P$ do not  belong to any of the above classes. More specifically, what makes current algorithms not readily applicable to Problem  $\cal P$ is the lack in the objective function $E(\mathbf{Q})$ of a(n additively) separable convex and nonconvex part [each $E_{i_n}(\mathbf{Q})$ in (\ref{MUs_energy}) is in fact the ratio of two functions, $\text{tr}(\bQ_{i_n})$ and $\Delta_{i_n}^{\texttt{t}}(\mathbf{Q}_{i_n},\mathbf{Q}_{-n}^\nu)$,  of the \emph{same} set of variables]. Therefore, the proposed approximation function $\tilde{E}(\mathbf{Z};\mathbf{Z}^\nu)$, along with the resulting SCA-algorithm, i.e., Algorithm 2, are an innovative  contribution of this work.  \vspace{-0.3cm}

\section{Distributed Implementation}
\label{section:decentralized}
To alleviate the communication overhead of a centralized
implementation (Algorithm 2), in this section we devise \emph{distributed} algorithms converging to local optimal solutions of  $\mathcal P$. Following \cite{Scutari_nonconvex}, the main idea is to choose the approximation functions $\tilde{E}$ and $\tilde{g}_{i_n}$ so that (on top of satisfying conditions P.1-P.3 and C.1-C.6, needed for convergence) the resulting  convexified problems $\mathcal P^\nu$ can be decomposed into (smaller) subproblems   solvable in parallel across  the SCeNBs, with limited signaling between the SCeNBs  and the  cloud.

Since the
approximation function $\tilde{E}$ introduced in (\ref{E_tilde}) is (sum) separable in the optimization variables of the MUs in each cell,   any  choice of  $\tilde{g}_{i_n}$'s enjoying the  same  decomposability structure   leads naturally to  convexified problems $\mathcal P^\nu$ that can be readily decomposed across the SCeNBs by using standard primal or dual decomposition techniques.

Of course  there is more than one choice of $\tilde{g}_{i_n}$ meeting the above requirements; all of them lead to \emph{convergent} algorithms that however differ for convergence speed, complexity, communication overhead, and a-priori knowledge of the system parameters. As case study, in the following, we consider  two representative valid approximants.   The first candidate $\tilde{g}_{i_n}$ is obtained exploiting the Lipschitz property of the gradient of the rate functions $r_{i_n}$, whereas the second one is based on  an equivalent reformulation of $\mathcal P$ introducing proper slack variables. The first choice offers a lot of flexibility in
the design of distributed algorithms$-$both primal and dual-based
schemes can be invoked$-$but it requires knowledge of all the
Lipschitz constants. The second choice does not need this knowledge, but it involves a higher computational cost at the SCeNBs side, due to the presence of the slack variables.% and less flexibility in the algorithmic design$-\mathcal P^\nu$ decouples only in the dual domain. We describe in detail the two proposed approximations next.
\vspace{-0.2cm}
 \subsection{{Per-cell distributed dual and primal decompositions}}
The approximation function $\tilde{g}_{i_n}$  in (\ref{g_tilde})   has the desired property of preserving   the structure of the original constraint function $g_{i_n}$ ``as much as possible'' by keeping the convex part  $r_{i_n}^+(\bQ)$ of $r_{i_n}(\bQ)$ unaltered. Numerical results show that this choice leads to fast convergence schemes, see Sec. \ref{sec:num_res}.  However the structure of $\tilde{g}_{i_n}$  prevents  $\mathcal P^\nu$ to  be decomposed across the SCeNBs due to   the \emph{nonadditive} coupling among the variables $\bQ_n$ in $r_{i_n}^+(\bQ)$.
 To cope with this issue, we   lower bound $r_{i_n}^+(\bQ)$ [and thus upper bound $\tilde{g}_{i_n}$  in (\ref{g_tilde})], so that we  obtain an alternative   approximation of $g_{i_n}$ that is \emph{separable in all} the $\bQ_n$'s,  while still satisfying C.1-C.6.
Invoking the Lipschitz property of the (conjugate) gradients $\nabla_{\mathbf{Q}_{j_{l}}^{*}}r_{i_{n}}^{\,\text{+}}(\bullet)$  on $\mathcal{Q}$,
 with constant $L_{{j_l}, i_n}$ [given in (19)  in Appendix B of the supporting material], we have %(\ref{Lipsc_cost_r})
 \beq \label{rin_up}
 \begin{array}{lll}
 r_{i_n}^{\,\text{+}}(\mathbf{Q}) \geq \tilde{r}_{i_n}^{\,\text{+}}(\mathbf{Q};\mathbf{Q}^{\nu})\triangleq r_{i_n}^{\,\text{+}}(\mathbf{Q}^{\nu}) \smallskip\\ +\ds \sum_{j_l \in \mathcal{I}}\left(
 \left\langle \boldsymbol{\Pi}^{\,\text{+}}_{j_l,i_n}(\mathbf{Q}^{\nu}), \mathbf{Q}_{j_l}-\mathbf{Q}^{\nu}_{j_l}\right\rangle-\ds   c_{j_l, i_n} \parallel \mathbf{Q}_{j_l}-\mathbf{Q}^{\nu}_{j_l} \parallel^{2}\right),
\end{array}\nonumber
 \eeq
for all $\bQ, \bQ^\nu\in \mathcal Q$,   where each $\boldsymbol{\Pi}_{j_{l},i_{n}}^{\,\text{+}}(\mathbf{Q}^{\nu})$ and $c_{j_l, i_n}$ are defined respectively as
 \begin{equation}
 \boldsymbol{\Pi}_{j_{l},i_{n}}^{\,\text{+}}(\mathbf{Q}^{\nu})\triangleq \left\{ \begin{split}& \nabla_{\mathbf{Q}_{j_{l}}^{*}}r_{i_{n}}^{\,\text{+}}(\mathbf{Q}^{\nu}), & \text{if }l\neq n\,\,\,\text{or}\,\,\, j_{l}=i_{n},\\
& \mathbf{0}, & \text{otherwise}\qquad\qquad\,\,
\end{split}
\right.
 \end{equation}
 with $\!\nabla_{\mathbf{Q}_{j_{l}}^{*}}r_{i_{n}}^{\,\text{+}}(\mathbf{Q}^{\nu})\!\!=\!\!\mathbf{H}_{j_{l}n}^{H}(\mathbf{R}_{n}(\mathbf{Q}_{-n}^{\nu})+\!
 \mathbf{H}_{i_{n}n}\mathbf{Q}_{i_{n}}^{\nu}\mathbf{H}_{i_{n}n}^{H})^{-1}\mathbf{H}_{j_{l}n}$ and
  \beq
    c_{j_l, i_n}\triangleq  \left\{
    \begin{split}
    &\ds L_{{j_l}, i_n},\, & \text{if }l\neq n\,\,\,\text{or}\,\,\, j_{l}=i_{n},\\ &0,&\text{otherwise}.\qquad\qquad\,
  \end{split} \right.
  \eeq
Note that  $ \tilde{r}_{i_n}^{\,\text{+}}(\mathbf{Q};\mathbf{Q}^{\nu})$ is (sum) separable in the MUs' covariance matrices $\bQ_{i_n}$'s. The desired approximant of $g_{i_n}$ can be then  obtained just replacing  $r_{i_n}^+(\bQ)$ in $\tilde{g}_{i_n}$   with $ \tilde{r}_{i_n}^{\,\text{+}}(\mathbf{Q};\mathbf{Q}^{\nu})$ [cf. (\ref{g_tilde})], resulting in \vspace{-0.2cm}
 \beq \label{g_upper_bound}
\begin{array}{lll}
\tilde{q}_{i_{n}}(\mathbf{Q},f_{i_{n}};\mathbf{Q}^{\nu})\hspace{-0.3cm}& \triangleq -\tilde{r}_{i_{n}}^{\,\text{+}}(\mathbf{Q};\mathbf{Q}^{\nu})+\dfrac{c_{i_{n}}\cdot f_{i_{n}}}{f_{i_{n}}\cdot\tilde{T}_{i_{n}}-w_{i_{n}}}\medskip\\
  & \hspace{0.32cm} -r_{n}^{\,\text{-}}(\mathbf{Q}_{-n}^{\nu})\!-\!\!\!{\displaystyle {\sum_{j_{l}\in\mathcal{I}}}\!\!\left\langle \boldsymbol{\Pi}_{j_{l},n}^{\,\text{-}}(\mathbf{Q}^{\nu}),\mathbf{Q}_{j_{l}}-\mathbf{Q}_{j_{l}}^{\nu}\right\rangle }\medskip\\
\hspace{-0.3cm} & \triangleq   \ds\sum_{j_{l}\in\mathcal{I}}\tilde{q}_{j_{l},i_{n}}(\mathbf{Q}_{j_{l}};\mathbf{Q}^{\nu})+\bar{q}_{i_n}(f_{i_n};\mathbf{Q}^{\nu})
\end{array}
\eeq
with $\tilde{q}_{j_{l},i_{n}}(\mathbf{Q}_{j_{l}};\mathbf{Q}^{\nu})$ and $\bar{q}_{i_n}(f_{i_n};\mathbf{Q}^{\nu})$  given by
   \beq
   \begin{array}{lll}
 \begin{split}
    \! \tilde{q}_{j_l,i_n}(\mathbf{Q}_{j_l}; \mathbf{Q}^{\nu})& \triangleq
 \ds c_{j_l, i_n}  \parallel \mathbf{Q}_{j_l}-\mathbf{Q}^{\nu}_{j_l} \parallel^{2}\\
   &-\left\langle \boldsymbol{\Pi}^{\,\text{+}}_{j_l,i_n}(\mathbf{Q}^{\nu})+
 \boldsymbol{\Pi}^{\,\text{-}}_{j_l,n}(\mathbf{Q}^{\nu}), \mathbf{Q}_{j_l}-\mathbf{Q}^{\nu}_{j_l}\right\rangle,\\
 \end{split}\medskip\\
  \bar{q}_{i_n}(f_{i_n}; \mathbf{Q}^{\nu}) \triangleq \dfrac{c_{i_{n}}\cdot f_{i_{n}}}{f_{i_{n}}\cdot\tilde{T}_{i_{n}}-w_{i_{n}}}
  -{r}_{i_n}(\mathbf{Q}^{\nu}). \end{array}\nonumber
 \eeq
It is not difficult to check that $\tilde{q}_{i_{n}}(\mathbf{Q},f_{i_{n}};\mathbf{Q}^{\nu})$, on top  of being separable in the MUs' covariance matrices,   also satisfies the required conditions C.1-C.6.  Using $\tilde{q}_{i_{n}}(\mathbf{Q},f_{i_{n}};\mathbf{Q}^{\nu})$ instead of $\tilde{g}_{i_{n}}(\mathbf{Q},f_{i_{n}};\mathbf{Q}^{\nu})$, the convexified subproblem replacing $\mathcal P^{\nu}$ is:
 given $\mathbf{Z}^{\nu} \in \mathcal{X}$,
 \begin{equation}
 \begin{array}{lll}
\!\!\widehat{\mathbf{Z}}(\mathbf{Z}^\nu)\triangleq\!\!\!  \begin{array}[t]{clll}
 \underset{\mathbf{Q},  \mathbf{f}}{ \text{argmin}}
\!\!& \!\! \,\displaystyle{\sum_{i_n \in \mathcal{I}}}\, \tilde{E}_{i_n}(\mathbf{Z}_{i_n};\mathbf{Z}^\nu)\vspace{-2cm}\\
\!\!\mbox{ s.t.} \!\!& \!\!    \begin{array}{llll} {}\vspace{1.8cm}\\ \!\! \texttt{a)}\,\,\!\!\!\!
\ds\sum_{j_{l}\in\mathcal{I}}\tilde{q}_{j_{l},i_{n}}(\mathbf{Q}_{j_{l}};\mathbf{Q}^{\nu})+\bar{q}_{i_n}(f_{i_n};\mathbf{Q}^{\nu}) \leq 0,\\
 \hspace{4cm}\forall  i_n\in \mathcal I,\medskip \\ \!\! \texttt{b)}\,\,\!\!\!\! \displaystyle{\sum_{i_n\in \mathcal I}}f_{i_n}\leq f_T,\quad f_{i_n}\geq 0,\,\, \forall  i_n\in \mathcal I,\medskip\\ \!\! \texttt{c)}\,\,\!\! \mathbf{Q}_{i_n}\in \mathcal{Q}_{i_n},\quad \forall i_n\in \mathcal I,\end{array}
\end{array} \label{Lip_problem}\tag{$\mathcal P_d^\nu$}%\vspace{+0.1cm}
\end{array}
\end{equation}
where with a slight abuse of notation we still use $\widehat{\mathbf{Z}}(\mathbf{Z}^\nu)\triangleq (\widehat{\mathbf{Q}}(\mathbf{Z}^\nu),\widehat{\mathbf{f}}(\mathbf{Z}^\nu))$ to denote the unique solution of \ref{Lip_problem}.  \\
\indent Problem \ref{Lip_problem} is now (sum) separable in the MUs' covariance matrices; it can be solved in  a distributed way using standard primal or dual decomposition techniques. We briefly show next how to customize standard  dual algorithms to \ref{Lip_problem}.

\subsubsection{Per-cell optimization via dual decomposition}

The subproblems \ref{Lip_problem} can be solved in a distributed way
if the side constraints $\tilde{q}_{i_{n}}(\mathbf{Q},f_{i_{n}};\mathbf{Q}^{\nu})\leq 0$  are dualized (note that there is zero duality gap). The dual problem associated with \ref{Lip_problem} is: given $\mathbf{Z}^{\nu}\triangleq ( \mathbf{Q}^{\nu},  \mathbf{f}^{\nu})\in {\mathcal{X}}$,
\beq
\underset{\boldsymbol{\lambda}\triangleq ((\lambda_{i_n})_{i_n \in \mathcal{I}}, \lambda_f)\geq \mathbf{0}}\max \quad D\left(\hat{\mathbf{Z}}(\boldsymbol{\lambda}; \mathbf{Z}^{\nu}), \boldsymbol{\lambda}; \mathbf{Z}^{\nu}\right) \label{distr_dual}
\eeq
where  $\hat{\mathbf{Z}}(\boldsymbol{\lambda};\bZ^\nu)\triangleq (\hat{\mathbf{Z}}_{n}(\boldsymbol{\lambda};\mathbf{Z}^{\nu}))_{n=1}^{N_c}$, with each $\hat{\mathbf{Z}}_{n}(\boldsymbol{\lambda};\mathbf{Z}^{\nu})\triangleq (\hat{\mathbf{Q}}_{n}(\boldsymbol{\lambda};\mathbf{Z}^{\nu}),\hat{\mathbf{f}}_{n}(\boldsymbol{\lambda};\mathbf{Z}^{\nu}))=(\hat{\mathbf{Q}}_{i_n}(\boldsymbol{\lambda};\mathbf{Z}^{\nu}),\hat{\mathbf{f}}_{i_n}(\boldsymbol{\lambda};\mathbf{Z}^{\nu}))_{i=1}^{K_n}$, is the unique minimizer of the Lagrangian function associated with \ref{Lip_problem}, which after reorganizing terms can be written as\vspace{-0.1cm}
\beq \vspace{-0.2cm}
\begin{array}{llll}
\hspace{-0.21cm}\hat{\mathbf{Z}}(\boldsymbol{\lambda};\bZ^\nu)\!\!\triangleq\hspace{-0.68cm}
\underset{\quad\quad\mathbf{Q} \in \mathcal{Q}, \mathbf{f} \in \mathbb{R}^{|\mathcal{I}|}_{+}}{\text{argmin}}  &\!\!\!\!\! \ds \sum_{n=1}^{N_c}\left(\mathcal{L}_{\mathbf{Q}_n}(\mathbf{Q}_n, \boldsymbol{\lambda}; \mathbf{Q}^{\nu} )\!+
    \mathcal{L}_{\mathbf{f}_n}(\mathbf{f}_n, \boldsymbol{\lambda};\mathbf{f}^{\nu}_n )\right)\!,
 \end{array} \label{dual_distr}
\eeq where $\mathbf{Q}_{n}\triangleq (\mathbf{Q}_{i_n})_{i=1}^{K_n}$, $\mathbf{f}_n\triangleq (f_{i_n})_{i=1}^{K_n}$, and

\beq \label{dual_Lagr}
\begin{split}
& \begin{split}
 & \mathcal{L}_{\mathbf{Q}_{n}}(\mathbf{Q}_{n}, \boldsymbol{\lambda};\mathbf{Q}^{\nu} )=\\ &\qquad \quad\ds \sum_{i=1}^{K_n} \left\{\tilde{E}_{i_n}(\mathbf{Q}_{i_n}, f_{i_n}^{\nu};\mathbf{Z}^{\nu}) +  \ds  \sum_{j_l \in \mathcal{I}}\lambda_{j_l} \tilde{q}_{i_n, j_l}(\mathbf{Q}_{i_n}; \mathbf{Q}^{\nu})\right\}\!,
\end{split} \\
& \mathcal{L}_{\mathbf{f}_{n}}(\mathbf{f}_{n}, \boldsymbol{\lambda};\mathbf{f}_{n}^{\nu})\!= \!\!\ds \sum_{i=1}^{K_n}\!\left\{\!\frac{c_f}{2} (f_{i_n}-f^{\nu}_{i_n})^2\!+\! \dfrac{\lambda_{i_n} \cdot c_{i_n}\cdot f_{i_n}}{f_{i_n}\cdot \tilde{T}_{i_n}-\omega_{i_n}}\!+ \!\! \lambda_f  f_{i_n}\!\!\right\}\!.\bigskip\\
\end{split}
\eeq

 \begin{algorithm}[t]
\textbf{Initial data:} $\boldsymbol{\lambda}^{0}\geq \mathbf{0}$, $\mathbf{Z}^{\nu}=(\mathbf{Q}^{\nu},\mathbf{f}^{\nu})$, $\{\beta_{k}\}>0$. Set $k=0$,

(\texttt{S.1}): If $\boldsymbol{\lambda}^{k}$  satisfies a suitable termination criterion:\texttt{STOP};

(\texttt{S.2}):  For each SCeNB  $n$, compute in parallel ${\mathbf{Q}}^{k+1}_n(\boldsymbol{\lambda}^{k};\mathbf{z}^{\nu})$
and ${\mathbf{f}}^{k+1}_n(\boldsymbol{\lambda}^{k};\mathbf{z}^{\nu})$ [cf.  (\ref{dual_distr_nuf})];

(\texttt{S.3}):  Update at the master node $\boldsymbol{\lambda}^{k+1}$ according to\medskip

 $ $  ${\lambda}^{k+1}_{i_n}\triangleq \left[{\lambda}^{k}_{i_n} +\beta_{k} \left(\ds\sum_{j_{l}\in\mathcal{I}}\tilde{q}_{j_{l},i_{n}}(\mathbf{Q}_{j_{l}};\mathbf{Q}^{\nu})+\bar{q}_{i_n}(f_{i_n};\mathbf{Q}^{\nu})\right) \right]^{+}$,

 \hspace{7.4cm}$\forall i_n \in \mathcal{I}\medskip$

  $${\lambda}^{k+1}_{f}\triangleq \left[{\lambda}^{k}_{f} +\beta_{k} \left( \ds \sum_{i_n \in \mathcal{I}} f^{k+1}_{i_n}-f_T \right)\right]^{+}$$

(\texttt{S.4}): $k\leftarrow k+1$ and go back to (\texttt{S.1}).

\caption{\textbf{:}   Distributed implementation of  S.2   in Alg. \ref{alg:Alg_centr}. \label{alg:Alg_dual_distr}}

\end{algorithm}

Note that, thanks to the separability structure of the Lagrangian function,    the optimal solutions  $\hat{\mathbf{Z}}_{n}(\boldsymbol{\lambda};\mathbf{Z}^{\nu})=(\hat{\mathbf{Q}}_n(\boldsymbol{\lambda};\mathbf{Q}^{\nu}),\hat{\mathbf{f}}_n(\boldsymbol{\lambda};\mathbf{f}^{\nu}))$  of (\ref{dual_distr}) can be   computed in parallel across the  SCeNBs, solving each  SCeNBs $n$ the following strongly convex problems:  given $\boldsymbol{\lambda}\geq \mathbf{0}$,
  \beq
\begin{array}{llll}
\hat{\mathbf{Q}}_n(\boldsymbol{\lambda};\mathbf{Q}^{\nu})&\triangleq &
\underset{\mathbf{Q}_{n} \in \, \Pi_{i=1}^{K_n} \mathcal{Q}_{i_n}}{\mbox{argmin}}
\left\{  \mathcal{L}_{\mathbf{Q}_{n}}(\mathbf{Q}_{n}, \boldsymbol{\lambda};\mathbf{Q}^{\nu} ) \right\}\medskip\\
  \hat{\mathbf{f}}_n(\boldsymbol{\lambda};\mathbf{f}^{\nu})&\triangleq &
\underset{ \mathbf{f}_n \in \mathbb{R}^{K_n}_{+}}{\mbox{argmin}}
\left\{  \mathcal{L}_{\mathbf{f}_{n}}(\mathbf{f}_{n}, \boldsymbol{\lambda};\mathbf{f}_{n}^{\nu} ) \right\}.
 \end{array} \label{dual_distr_nuf}
\eeq

 The solution of \ref{Lip_problem}  can be then computed solving the dual problem (\ref{distr_dual}). It is not difficult to prove that the   dual function $D$ is  differentiable with Lipschitz gradient. One can then solve (\ref{distr_dual})  using, e.g.,
  the gradient-based algorithm with diminishing step-size
  described  in Algorithm \ref{alg:Alg_dual_distr}, whose convergence is stated in  Theorem \ref{thm:dual_decomp} (the proof follows standard arguments and thus is omitted, because of space limitations).\vspace{-0.1cm}

\begin{theorem} \label{thm:dual_decomp} {\it
Given   \ref{Lip_problem},
choose  $\{\beta_k\}$ so that $\beta_k> 0$, $\beta_{k}\rightarrow 0$, $\sum_{k}\beta_{k}=+\infty$, and $\sum_{k}(\beta_{k})^2<\infty$.
Then,
the sequence $\{\boldsymbol{\lambda}_k\}$  generated by Algorithm \ref{alg:Alg_dual_distr} converges to a solution of
 (\ref{distr_dual}). Therefore,   the sequence
 $\{\hat{\mathbf{Z}}^k(\boldsymbol{\lambda}_k;\mathbf{Z}^{\nu})\}_k$
 converges to the unique solution of \ref{Lip_problem}. } \hfill $\square$
\end{theorem}\vspace{-0.3cm}

 \subsection{Alternative decomposition via slack variables}

 In this section we present an alternative decomposition strategy of problem $\mathcal P$ that does not require the knowledge of the Lipschitz constants $L_{{j_l}, i_n}$. At the basis of our approach there is an equivalent reformulation of $\mathcal P$ based on the introduction of proper slack variables that are instrumental  to decouple in each $r_{i_n}^+(\bQ)$ [cf. \eqref{r_lus}] the covariance matrix $\bQ_{i_n}$ of user $i_n$ from those of the  MUs in the other cells$-$the interference term $\mathbf{R}_{n}(\bQ_{-n})$ [cf. (\ref{eq:MUI})]. More specifically, introducing the slack variables $\bY_{i_n}$,  and
\beq
\mathbf{I}_{i_n}(\mathbf{Q})  \triangleq \sum_{j_m\in \mathcal I, m\neq n}\mathbf{H}_{j_m n} \mathbf{Q}_{j_m} \mathbf{H}_{j_m n}^H+ \mathbf{H}_{i_n n} \mathbf{Q}_{i_n} \mathbf{H}_{i_n n}^H,
\eeq
we can write  \vspace{-0.2cm}
\beq\label{r_plus_tilde}
r_{i_n}^+(\bQ)=\overline{r}_{i_n}^{\,+}(\bY),\vspace{-0.2cm}
\eeq
with
\beq \label{Y_slack}
\overline{r}_{i_n}^{\,+}(\bY) \triangleq \log_2 \det \left(\mathbf{R}_{w}+\mathbf{Y}_{i_n}\right)\,\, \text{and}\,\, \mathbf{Y}_{i_n}=\mathbf{I}_{i_n}(\mathbf{Q}).
\eeq
 Using \eqref{r_plus_tilde}, \eqref{Y_slack}, and $g_{i_{n}}(\mathbf{Q},f_{i_{n}})$ written as in   \eqref{g_DC},
the original offloading problem  $\mathcal P$ can be rewritten in the  following equivalent form: denoting $\bY\triangleq  (\bY_{i_n})_{i_n\in \mathcal I}$,
 \beq\label{P_tilde}
\hspace{-0.3cm}\begin{array}{l}
\underset{\mathbf{Q},\mathbf{f},\mathbf{Y}}{\min}\,\, {E}(\mathbf{Q})\\
\begin{array}{ll}
\mbox{\!\! s.t.} & \texttt{a)}\,-\overline{r}_{i_n}^{\,+}(\mathbf{Y}_{i_{n}})-r_{n}^{\,\text{-}}(\mathbf{Q}_{-n})+\frac{c_{i_{n}}\cdot f_{i_{n}}}{f_{i_{n}}\cdot\tilde{T}_{i_{n}}-w_{i_{n}}}\leq0,\,\forall{i_{n}\in\mathcal{I}},\medskip\\
 & \texttt{b)}\,{\displaystyle {\sum_{i_{n}\in\mathcal{I}}}f_{i_{n}}\leq f_{T},\quad f_{i_{n}}\geq0,\,\,\,\,\, \forall{i_{n}\in\mathcal{I}},\medskip}\\
 & \texttt{c)}\,\mathbf{Q}_{i_{n}}\in\mathcal{Q}_{i_{n}},\hspace{2.3cm}\forall{i_{n}\in\mathcal{I}},\medskip\\
 & \texttt{d)}\,\mathbf{0}\preceq \mathbf{Y}_{i_{n}}\preceq\mathbf{I}_{i_{n}}(\mathbf{Q}),\hspace{0.6cm}\forall{i_{n}\in\mathcal{I}}.
\end{array}\vspace{-0.2cm}
\end{array}\tag{$\tilde{\mathcal P}$}
\eeq
 We denote by $\tilde{\mathcal X}$ the feasible set of  $\tilde{\mathcal P}$. The equivalence between ${\mathcal P}$ and $\tilde{\mathcal P}$ is stated next.
 \begin{lemma} Given the nonconvex problems ${\mathcal P}$ and $\tilde{\mathcal P}$, the following hold:

 \noindent (a): Every  {feasible} point of $\tilde{\mathcal P}$ (or ${\mathcal P}$) is regular (i.e., satisfies the Mangasarian-Fromovits Constraint Qualification \cite{Facchinei-Pang_FVI03});

  \noindent (b): ${\mathcal P}$ and $\tilde{\mathcal P}$ are equivalent in the following sense. If $(\bar{\bQ},\bar{\mathbf{f}})$ is a stationary solution of ${\mathcal P}$, then there exists a  $\bar{\bY}$ such that  $(\bar{\bQ},\bar{\mathbf{f}},\bar{\bY})$ is a stationary solution of $\tilde{\mathcal P}$; and   viceversa. \hfill $\square$
  \end{lemma}%\vspace{-0.2cm}

 Condition (a) in the lemma guarantees the existence of stationary points of $\tilde{\mathcal P}$, whereas (b) allows us to compute (stationary) solutions of $\mathcal P$ solving \ref{P_tilde}.

  We convexify next \ref{P_tilde} following the same guidelines as   in Sec. \ref{section:offloading_multiple_cells} [see P.1-P.3 and C.1-C.6]. Introducing \vspace{-0.2cm}
\begin{equation}\label{g_tildey}
\begin{array}{lll}
\tilde{g}_{i_n}(\mathbf{Q}, f_{i_n}, \mathbf{Y}_{i_n};\mathbf{Q}^\nu)\triangleq
- \overline{r}_{i_n}^{\,+}(\mathbf{Y}_{i_n}) + \dfrac{c_{i_n}\cdot f_{i_n}}{f_{i_n}\cdot \tilde{T}_{i_n}-w_{i_n}} \medskip\\
- r_{n}^{\,\text{-}}(\mathbf{Q}_{-n}^\nu) -  \displaystyle{\sum_{j_m \in \mathcal{I}}} \left\langle \boldsymbol{\Pi}^{\,\text{-}}_{j_m,n}(\mathbf{Q}^{\nu}), \mathbf{Q}_{j_m}-\mathbf{Q}_{j_m}^{\nu} \right\rangle,
\end{array}
\end{equation}
 and using the same approximant $\tilde{E}(\bZ;\bZ^\nu)$ as defined in (\ref{E_tilde_in}), we have:  given a feasible $\mathbf{W}^\nu\triangleq (\textbf{Z}^\nu, \bY^\nu)$,
\begin{equation}
 \begin{array}{ll}
\hat{\mathbf{W}}(\mathbf{W}^{\nu})\triangleq &  \!\!\!\!\underset{\mathbf{Q},\mathbf{f},\mathbf{Y}}{{\text{argmin}}}\!\!\quad\tilde{E}(\mathbf{Z};\mathbf{Z}^{\nu})+\dfrac{c_{\mathbf{Y}}}{2}\|\mathbf{Y}-\mathbf{Y}^{\nu}\|^{2}\\
& \begin{array}{ll}
\!\!\mbox{\!\! s.t.} & \texttt{a)}\,\tilde{g}_{i_{n}}(\mathbf{Q},,f_{i_{n}},\mathbf{Y}_{i_{n}};\mathbf{Q}^{\nu})\leq 0,\,\,\forall{i_{n}\in\mathcal{I}},\medskip\\
 & \texttt{b)}\,{\displaystyle {\sum_{i_{n}\in\mathcal{I}}}f_{i_{n}}\leq f_{T},\quad f_{i_{n}}\geq0,\quad\forall{i_{n}\in\mathcal{I}},\medskip}\\
 & \texttt{c)}\,\mathbf{Q}_{i_{n}}\in\mathcal{Q}_{i_{n}},\quad\forall{i_{n}\in\mathcal{I}},\medskip\\
 & \texttt{d)}\,\mathbf{0}\preceq\mathbf{Y}_{i_{n}}\preceq\mathbf{I}_{i_{n}}(\mathbf{Q}),\forall{i_{n}\in\mathcal{I}}
\end{array}
\end{array}
 \label{prob_cvx_y}\tag{$\widetilde{{P}}^{\nu}$}
\end{equation}
where $\hat{\mathbf{W}}(\mathbf{W}^{\nu})=(\hat{\mathbf{Q}}(\mathbf{W}^{\nu}),\hat{\mathbf{f}}(\mathbf{W}^{\nu}),\hat{\mathbf{Y}}(\mathbf{W}^{\nu}))$ denotes  the unique solution of \ref{prob_cvx_y}, and $c_{\mathbf{Y}}$ is an arbitrary positive constant.\\
\indent The stationary solutions of $\tilde{\mathcal P}$ (and thus ${\mathcal P}$) can be computed solving the sequence of strongly convex problems \ref{prob_cvx_y}. The formal description of the scheme is still given by  Algorithm 2 wherein in Step 2,  $\hat{\mathbf{Z}}(\mathbf{Z}^{\nu})$ is replaced by  $\hat{\mathbf{W}}(\mathbf{W}^{\nu})$; convergence is guaranteed under conditions in Theorem 2.\\
\indent The last thing left is   showing how to solve each subproblem  \ref{prob_cvx_y} in a distributed way.
Problem \ref{prob_cvx_y} can be decoupled across the SCeNB's in the dual domain (note that there is zero duality gap). Indeed, denoting by   $\mathbf{W} \triangleq (\mathbf{Q}, \mathbf{f}, \mathbf{Y})$,  and $\boldsymbol{\lambda}\triangleq((\lambda_{i_n})_{i_n \in \mathcal{I}}, \lambda_f)$ and $\boldsymbol{\Omega}\triangleq (\boldsymbol{\Omega}_{i_n}\succeq
\mathbf{0})_{i_n \in \mathcal{I}}$   the multipliers associated with the constraints (a), (b), and (d), respectively, the (partial) Lagrangian  has the following \emph{additive} structure:
\beq\nonumber
\hspace{-0.2cm}\begin{array}{llll}
\begin{split}
\mathcal{L}(\mathbf{W}, \boldsymbol{\lambda}, \boldsymbol{\Omega}; \mathbf{Q}^{\nu})& \triangleq
\ds \sum_{n=1}^{N_c}\left\{\mathcal{L}_{\mathbf{Q}_n}(\mathbf{Q}_n,\boldsymbol{\lambda}, \boldsymbol{\Omega}; \mathbf{Y}^{\nu})+ \right.\\  & \left. \mathcal{L}_{\mathbf{Y}_n}(\mathbf{Y}_n, \boldsymbol{\lambda}, \boldsymbol{\Omega}; \mathbf{W}^{\nu})+
\mathcal{L}_{\mathbf{f}_{n}}(\mathbf{f}_{n}, \boldsymbol{\lambda},\mathbf{f}_{n}^{\nu} )\right\},\end{split}
\end{array}
\eeq
 where\vspace{-0.2cm}
\beq\nonumber
\begin{array}{llllll}\hspace{-0.22cm}
\begin{split}
\mathcal{L}_{\mathbf{Q}_n}(\mathbf{Q}_n,\boldsymbol{\lambda}, \boldsymbol{\Omega}; \mathbf{W}^{\nu})\!=& \!\ds \sum_{i=1}^{K_n} \left\{\tilde{E}_{i_n}(\mathbf{Q}_{i_n}, f_{i_n}^{\nu};\mathbf{Z}^{\nu}) - \! \lambda_{i_n} r_{n}^{\,\text{-}}(\mathbf{Q}_{-n}^\nu) \right. \\
 & -\!\! \displaystyle{\sum_{j_m \in \mathcal{I}}} \lambda_{j_m}  \left\langle \boldsymbol{\Pi}^{\,\text{-}}_{i_n,j_m}(\mathbf{Q}^{\nu}), \mathbf{Q}_{i_n}-\mathbf{Q}_{i_n}^{\nu} \right\rangle\\
& - \!\!\sum_{j_{m} \in \mathcal{I}, m\neq n}
 \left\langle \boldsymbol{\Omega}_{j_m}, \mathbf{H}_{i_n m}\mathbf{Q}_{i_n} \mathbf{H}_{i_n m}^{H} \right\rangle \\
 &\left. - \left\langle \boldsymbol{\Omega}_{i_n}, \mathbf{H}_{i_n n}\mathbf{Q}_{i_n} \mathbf{H}_{i_n n}^{H} \right\rangle \right\},
  \end{split}
  \\
  \begin{split} \mathcal{L}_{\mathbf{Y}_n}(\mathbf{Y}_n, \boldsymbol{\lambda}, \boldsymbol{\Omega}; \mathbf{W}^{\nu})= & \ds \sum_{i=1}^{K_n} \left\{- \lambda_{i_n}
 r_{i_n}^{\,\text{+}}(\mathbf{Y}_{i_n})+\left\langle \boldsymbol{\Omega}_{i_n}, \mathbf{Y}_{i_n} \right\rangle
 \right.\\
 & \left.+ \ds \frac{c_{\mathbf{Y}}}{2} \| \mathbf{Y}_{i_n}- \mathbf{Y}_{i_n}^{\nu}\|^2 \right\},\\
 \end{split}\\
\end{array}
\eeq
and $\mathcal{L}_{\mathbf{f}_{n}}(\mathbf{f}_{n}, \boldsymbol{\lambda},\mathbf{f}_{n}^{\nu} )$ is given by  (\ref{dual_Lagr}).
The minimization of $\mathcal{L}(\mathbf{W}, \boldsymbol{\lambda}, \boldsymbol{\Omega}; \mathbf{W}^{\nu})$ w.r.t. $\mathbf{W}=(\mathbf{Q}, \mathbf{f}, \mathbf{Y})\triangleq ( \mathbf{Q}_n, \mathbf{f}_n, \mathbf{Y}_n)_{n=1}^{N_c}$ becomes then
\beq
\hspace{-0.6cm}\begin{array}{llll}
& D(\boldsymbol{\lambda}, \boldsymbol{\Omega}; \mathbf{W}^{\nu})\triangleq  \ds \sum_{n=1}^{N_c} \left( \underset{\mathbf{Q}_n \in \mathcal{Q}}\min \mathcal{L}_{\mathbf{Q}_n}(\mathbf{Q}_n,\boldsymbol{\lambda}, \boldsymbol{\Omega}; \mathbf{W}^{\nu})\right.\\ & \hspace{0.6cm} \left.
+\!\!\! \underset{(\mathbf{Y}_{i_n}\succeq \mathbf{0})_{i_n \in \mathcal{I}}}\min  \!\!\!\!\mathcal{L}_{\mathbf{Y}_n}(\mathbf{Y}_n, \boldsymbol{\lambda}, \boldsymbol{\Omega}; \mathbf{W}^{\nu})+ \underset{ \mathbf{f} \in \mathbb{R}^{|\mathcal{I}|}_{+}}\min
\mathcal{L}_{\mathbf{f}_{n}}(\mathbf{f}_{n}, \boldsymbol{\lambda},\mathbf{f}_{n}^{\nu})\right)
 \end{array} \label{dual_distr_y1}
\eeq
whose unique solutions $\hat{\mathbf{W}}(\boldsymbol{\lambda}, \boldsymbol{\Omega};\mathbf{W}^{\nu})\triangleq (\hat{\mathbf{Q}}_n(\boldsymbol{\lambda},\boldsymbol{\Omega};\mathbf{Q}^{\nu}), $ $\hat{\mathbf{Y}}_n(\boldsymbol{\lambda}, \boldsymbol{\Omega};\mathbf{Y}^{\nu}),\hat{\mathbf{f}}_n(\boldsymbol{\lambda};\mathbf{f}^{\nu}))_{n=1}^{N_c}$ can be computed in parallel across the SCeNBs $n$:

  \beq
\begin{array}{l}
\hspace{-1.8cm}\hat{\mathbf{Q}}_n(\boldsymbol{\lambda}, \boldsymbol{\Omega};\mathbf{Q}^{\nu})\triangleq
\underset{\mathbf{Q}_{n}\in \mathcal{Q}_{n}}{\mbox{argmin}}
\left\{  \mathcal{L}_{\mathbf{Q}_{n}}(\mathbf{Q}_{n}, \boldsymbol{\lambda}, \boldsymbol{\Omega};\mathbf{Q}^{\nu} ) \right\}
 \end{array} \label{dual_distr_nuQ_y}
\eeq
 \beq
\begin{array}{l}
\hspace{-0.6cm}\hat{\mathbf{Y}}_n(\boldsymbol{\lambda}, \boldsymbol{\Omega};\mathbf{Y}^{\nu})\triangleq
\underset{(\mathbf{Y}_{i_n}\succeq \mathbf{0})_{i=1}^{K_n}}{\mbox{argmin}}
\left\{  \mathcal{L}_{\mathbf{Y}_{n}}(\mathbf{Y}_{n}, \boldsymbol{\lambda}, \boldsymbol{\Omega};\mathbf{Y}^{\nu} ) \right\}
 \end{array} \label{dual_distr_nuY}
\eeq
\beq
\begin{array}{l}
\hspace{-3.2cm}\hat{\mathbf{f}}_n(\boldsymbol{\lambda};\mathbf{f}^{\nu})\triangleq
\underset{ \mathbf{f}_n \in \mathbb{R}^{K_n}_{+}}{\mbox{argmin}}
\left\{  \mathcal{L}_{\mathbf{f}_{n}}(\mathbf{f}_{n}, \boldsymbol{\lambda};\mathbf{f}_{n}^{\nu} ) \right\}.
 \end{array} \label{dual_distr_nuf_y}
\eeq

Interestingly, problem (\ref{dual_distr_nuY})  admits a closed form solution.%, as given next.
\begin{lemma}\label{lemma:closed_form}
Let $\mathbf{U}^{H}_{i_n} {\mathbf{D}}_{i_n} \mathbf{U}_{i_n}$ be the eigenvalue/eigenvector decomposition of $c_{\mathbf{Y}} \mathbf{Y}_{i_n}^{\nu}-\mathbf{\Omega}_{i_n}$,
with ${\mathbf{D}}_{i_n}=\text{diag}(({{d}}_{i_n,j})_{j=1}^{n_{R_n}})$. The optimal solution of problem (\ref{dual_distr_nuY})
is
\beq
\mathbf{Y}_{i_n}=\mathbf{U}_{i_n} \mathbf{D}_{\mathbf{Y}_{i_n}} \mathbf{U}^{H}_{i_n}
\eeq
with  $\mathbf{D}_{\mathbf{Y}_{i_n}} =\text{diag}(({y}_{i_n, j})_{j=1}^{n_{R_n}})$ given by
\beq
\begin{array}{lll}\nonumber
{y}_{i_n, j}=\left[-\left(\frac{\sigma^{2}_{w}}{2}-\frac{{{d}}_{i_n,j}}{2 c_{\mathbf{Y}}}\right)+\sqrt{\left(\frac{\sigma^{2}_{w}}{2}+\frac{{{d}}_{i_n,j}}{2 c_{\mathbf{Y}}}\right)^2+\frac{\lambda_{i_n}}{2 c_{\mathbf{Y}}}}\right]^{+}.
\end{array}
\eeq
\end{lemma}
\begin{proof}
See Appendix  C in the supporting material for the proof here omitted  for lack of space. %\ref{subsec: Lemma2}
\end{proof}

Given  $\hat{\mathbf{W}}(\boldsymbol{\lambda}, \boldsymbol{\Omega};\mathbf{W}^{\nu})$, the dual problem associated with \ref{prob_cvx_y} is
\beq
\underset{\boldsymbol{\lambda}\geq 0, (\boldsymbol{\Omega}_{i_n}\succeq \mathbf{0})_{i_n \in \mathcal{I}}}\max \quad D(\boldsymbol{\lambda}, \boldsymbol{\Omega}; \mathbf{W}^{\nu}), \label{distr_dual_y}
\eeq
with $D(\boldsymbol{\lambda}, \boldsymbol{\Omega}; \mathbf{W}^{\nu})$ defined in (\ref{dual_distr_y1}). It can be show that the dual function is $C^2$, with Hessian Lipschitz continuous with respect to $\mathbf{W}^\nu$ on $\mathcal X$. Then, the dual problem \eqref{distr_dual_y} can be solved using either first or second order methods. An instance of gradient-based schemes is given in Algorithm \ref{alg:Alg_dual_distr_y}, whose convergence is guaranteed under the same conditions as in the Theorem \ref{thm:dual_decomp}. In S.3, the symbol $[\mathbf{A}]_+$ denotes the Euclidean projection of the square matrix $\mathbf{A}$ onto the convex set of   positive semidefinite matrices (having the same size of $\mathbf{A}$).\\
\indent A faster algorithm solving the dual problem can be  readily obtained using second order information. It is sufficient to replace the update of the multipliers in Step 3 of Algorithm \ref{alg:Alg_dual_distr_y} with the following (convergence is still guaranteed by Theorem 3):
\beq\label{Newton_algo}
\begin{split}
&\lambda^{k+1}_{i_n}= \lambda^{k}_{i_n}+\beta_k(\hat{\lambda}^{k+1}_{i_n}-\lambda^{k}_{i_n}),\quad \forall i_n \in \mathcal{I}\\
&{\boldsymbol{\Omega}}^{k+1}_{i_n}\triangleq {\boldsymbol{\Omega}}^{k}_{i_n} +\beta_{k} ( \hat{\boldsymbol{\Omega}}^{k+1}_{i_n}-\boldsymbol{\Omega}^{k}_{i_n}), \forall i_n \in \mathcal{I} \\
&\lambda^{k+1}_{f}=\lambda^{k}_{f}+\beta_k(\hat{\lambda}^{k+1}_{f}-\lambda^{k}_{f})
\end{split}
\eeq
where
\beq \label{lambda_new}
\begin{split}
\hspace{-2.4cm}\hat{{\lambda}}^{k+1}_{i_n}\triangleq \left[\hat{{\lambda}}^{k}_{i_n} + (\nabla^2_{\lambda_{i_n}} D(\hat{\mathbf{W}}^{k+1},\boldsymbol{\lambda}, \boldsymbol{\Omega}; \mathbf{W}^{\nu}))^{-1} \right. \\ \cdot  \left. \nabla_{\lambda_{i_n}} D(\hat{\mathbf{W}}^{k+1},\boldsymbol{\lambda}, \boldsymbol{\Omega}; \mathbf{W}^{\nu})  \right]^{+},
\end{split}
\eeq
\beq \label{Omega_new}
\begin{split}
\text{vec}\!\left(\hat{{\boldsymbol{\Omega}}}^{k+1}_{i_n}\right)\! \triangleq &  \left[ \text{vec}\!\left(\hat{{\boldsymbol{\Omega}}}^{k}_{i_n} \right) \right.\!\!\!
+\!\!\left(  \nabla_{\text{vec}(\boldsymbol{\Omega}_{i_n}^{*})}^{2}D(\hat{\mathbf{W}}^{k+1},\boldsymbol{\lambda}, \boldsymbol{\Omega}; \mathbf{W}^{\nu})\right)^{-1}\\
 & \left. \cdot \text{vec}\left(  \nabla_{\boldsymbol{\Omega}_{i_n}^{*}}D(\hat{\mathbf{W}}^{k+1},\boldsymbol{\lambda}, \boldsymbol{\Omega}; \mathbf{W}^{\nu})\right)\right]_{+},
\end{split}
\eeq
\beq
\begin{split}
 \label{lambdaf_new} \hspace{-2.0cm}\hat{{\lambda}}^{k+1}_{f}\triangleq \left[\hat{{\lambda}}^{k}_{f} +(\nabla^{2}_{\lambda_{f}} D(\hat{\mathbf{W}}^{k+1},\boldsymbol{\lambda}, \boldsymbol{\Omega}; \mathbf{W}^{\nu}))^{-1}\cdot \right.\\     \left. \nabla_{\lambda_{f}} D(\hat{\mathbf{W}}^{k+1},\boldsymbol{\lambda}, \boldsymbol{\Omega}; \mathbf{W}^{\nu})\right]^{+}.
 \end{split}
\eeq

 \begin{algorithm}[h]
\textbf{Initial data:} $\boldsymbol{\lambda}^{0}\geq \mathbf{0}$, $\boldsymbol{\Omega}^{0} \succeq \mathbf{0}$, $\mathbf{W}^{\nu}=(\mathbf{Q}^{\nu},\mathbf{Y}^{\nu}, \mathbf{f}^{\nu})$, $\{\beta_{k}\}_k>0$. Set $k=0$,

(\texttt{S.1}): If $\boldsymbol{\lambda}^{k}$, $\boldsymbol{\Omega}^{k}$ satisfy a suitable termination criterion:\texttt{STOP};

(\texttt{S.2}):  For each SCeNB  $n$, compute in parallel  ${\mathbf{Q}}^{k+1}_n(\boldsymbol{\lambda}^{k};\boldsymbol{\Omega}^{k};\mathbf{W}^{\nu})$,
${\mathbf{Y}}^{k+1}_n(\boldsymbol{\lambda}^{k};\boldsymbol{\Omega}^{k};\mathbf{W}^{\nu})$ and ${\mathbf{f}}^{k+1}_n(\boldsymbol{\lambda}^{k};\mathbf{W}^{\nu})$   solving  (\ref{dual_distr_nuQ_y})-(\ref{dual_distr_nuf_y});

(\texttt{S.3}):  Update at the master node $\boldsymbol{\lambda}$ and  $\boldsymbol{\Omega}$  according to
$${\lambda}^{k+1}_{i_n}\triangleq \left[{\lambda}^{k}_{i_n} +\beta_{k} \tilde{g}_{i_n}({\mathbf{Q}}^{k+1}_{i_n},{\mathbf{Q}}^{k+1}_{-n},{{f}}^{k+1}_{i_n};\mathbf{Q}^{\nu}, {f}^{\nu}_{i_n}) \right]^{+},\,\,\forall i_n,$$

$${\lambda}^{k+1}_{f}\triangleq \left[{\lambda}^{k}_{f} +\beta_{k} \left( \ds \sum_{i_n \in \mathcal{I}} f^{k+1}_{i_n}-f_T \right)\right]^{+}$$

$${\boldsymbol{\Omega}}^{k+1}_{i_n}\triangleq \left[{\boldsymbol{\Omega}}^{k}_{i_n} +\beta_k \left( \mathbf{Y}^{k+1}_{i_n}-\mathbf{I}_{i_n}(\mathbf{Q}^{k+1})\right)\right]_+,\,\,\forall i_n \in \mathcal{I}$$

(\texttt{S.4}): $k\leftarrow k+1$ and go back to (\texttt{S.1}).

\caption{\textbf{:}  Distributed dual scheme solving  \ref{prob_cvx_y}  \label{alg:Alg_dual_distr_y}}

\end{algorithm}\vspace{-0.1cm}

The explicit expression of the  Hessian matrices and
gradients in (\ref{lambda_new})-(\ref{lambdaf_new}) is given in Appendix D in the supporting document and here omitted  for lack of space. Numerical results show that  using  second order information significantly enhances practical  convergence speed. \vspace{-0.5cm}
%\subsection{On the implementation of distributed algorithms}
%\textcolor{red}{TBD:The proposed distributed algorithms are based on multi-layer optimization techniques which can be implemented
%on two levels: first each SCeNB finds out the optimization variables associated to the served users in the cell  and
%then these variables are transmitted to the cloud for a centralized computation of  the  Lagrangian multipliers.
 %In particular in the dual-based and the slack-based approaches,
% the covariance matrices and the computational rates of the local MUs are calculated at a first cell level and then the optimal solutions are transferred to the cloud in order to find the optimal multipliers. This exchange is iterated until convergence is reached.
%In the primal-based approach instead the covariance matrices are found at the cell level, hence transmitted to the cloud which computes the optimal
%computational resources assignment and  these steps are iterated until convergence.
% Therefore being the proposed decentralized approaches based on the dual-problem optimization, they  require  the execution of two nested algorithms: this
%This implies an iterative exchange of
%data (Lagrangian multipliers and optimal solutions) among the small-cell radio access points and the cloud through the available backhaul.}
\section{Numerical results}\label{sec:num_res}

In this section we present   some numerical results  to  assess the effectiveness of the proposed joint optimization of the communication and computational resources.

The simulated scenario is the following. We consider a network composed of $N_c=2$ cells, where all transceivers are equipped with $n_T=n_R=2$ antennas (unless stated otherwise). In each cell, there are $K_n=4$ active users, randomly deployed. In all our experiments the system parameters are set as
  (unless stated otherwise):   $f_T=2\cdot 10^{7}$,  $\tilde{T}=0.1$, $w=10^{5}$, $\mathbf{R}_w=N_0 \mathbf{I}_{n_r} $, $\texttt{snr}=10$dB. This choice guarantees the nonemptiness of the feasible set $\mathcal X$;  the constant $\alpha$ in the diminishing step-size rule (\ref{eq:diminishing_step_size}) is chosen as  $\alpha=1$e$-4$,   and the termination accuracy $\delta$ is set to $10^{-3}$.

\noindent \emph{Example $\#$ 1: Joint vs. disjoint optimization}.
We start comparing the energy consumption of the proposed offloading strategy with a method where
communication and computational resources are optimized separately.
%The first one is based on the  following procedure: i) one first optimizes the users' covariance matrices (based on some optimality criterion) subject to the power constraints $\mathcal{Q}$; and ii) then  offloading is performed using as users' rates those resulting from the optimization in  step  i) (provided that the offloading is feasible).   As optimality criterion in i) we adopted the maximization of the weighted geometric mean, i.e.,  $\Pi_{i,n}r_{i_n}(\mathbf{Q})^{\bar{b}_{i_n}}$, with $\bar{b}_{i_n}=b_{i_n}/ \sum_{i,n}b_{i_n}$ (the idea is to assign higher priorities to users having more bits to transfer to the cloud), resulting in the following optimization problem:
% \begin{equation}
% \begin{array}{clll}
%\underset{\mathbf{Q}}{\max}
%& \displaystyle{\sum_{i_n\in \mathcal I}}{\bar{b}_{i_n}}\cdot  \log \left(r_{i_n}(\mathbf{Q})\right)\smallskip\\
%\mbox{s.t.}  &    \mathbf{Q}_{i_n}\in \mathcal{Q}_{i_n},\quad \forall {i_n\in \mathcal I} .
%\end{array}  \label{geo_mean}
%\end{equation}
%    The above problem is nonconvex; we used the SCA-based algorithm proposed in \cite{scutari_facchinei_et_al_tsp13}  to compute a local optimal solution of \eqref{geo_mean}. The resulting rates are then used to support MUs' offloading (when it is feasible). We will term this procedure \emph{PHY Optimization (PHYO)}  algorithm.
The benchmark used to assess the relative merits of our approach is an instance of Algorithm 2 wherein
the computational rates $f_{i_n}$ are not optimized but set proportional to the computational load of each user, while meeting the computational rate constraint $f_T$ with equality, i.e., $f_{i_n}=w_{i_n }f_T/\sum_{i_n \in \mathcal{I}} w_{i_n } $ CPU cycles/second.  We termed such a method  \emph{Disjoint Resource Allocation (DRA)} algorithm. Note that this algorithm is still guaranteed to converge by Theorem \ref{thm:centralized}.
An important parameter useful to assess the usefulness of offloading algorithms is the ratio $\eta_{i_n}:=w_{i_n}/b_{i_n}$ between the computational load $w_{i_n}$ to be transferred and the number of bits $b_{i_n}$ enabling the transfer.
Fig.  \ref{figheur5n} shows an example of overall energy consumption, assuming the same ratio $\eta_{i_n}:=\eta$ for all users, obtained using Algorithm 2 and DRA algorithm. In particular, $\eta$ is varied keeping a fixed work load $w$ and changing the number $b_{i_n}$ of bits to be sent. The radio channels are Rayleigh fading and the results  are averages over   $100$ independent channel realizations. %The system parameters have been chosen so that offloading is feasible for all the algorithms over the simulated channel realizations.
\begin{figure}[h]
 \centering
        %\begin{subfigure}[b]{0.28\textwidth}
        %         \includegraphics[width=1.7in,height=1.7in]{figheur5n.ps}
        %\end{subfigure}%
%        \begin{subfigure}[b]{0.28\textwidth}
        \includegraphics[width=7.3cm]{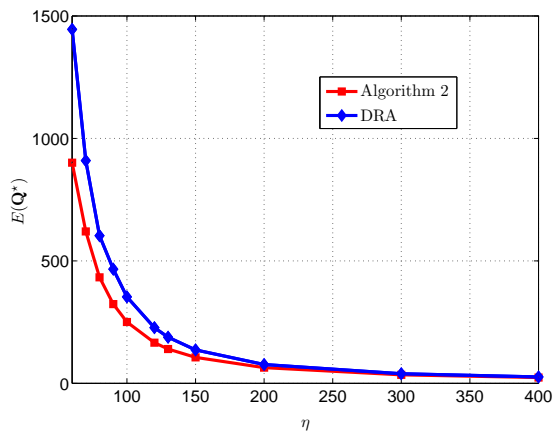}
        %\end{subfigure}\vspace{-0.3cm}
       \caption{Energy consumption vs.  $\eta=w_{i_n}/b_{i_n}$ for Algorithm \ref{alg:Alg_centr} and for DRA.}\label{figheur5n} \vspace{-0.2cm}
\end{figure}
 %\begin{figure}[t]
%\centering
%\includegraphics[width=7cm]{figheur3.eps}\vspace{-0.3cm}
%\caption{Energy consumption vs.  $\eta=w_{i_n}/b_{i_n}$: Proposed scheme (Algorithm \ref{alg:Alg_centr}), PHYO, and WCA.}\label{figheur}\vspace{-0.4cm}
%\end{figure}
Fig. \ref{figheur5n} shows a few interesting features: i) the joint optimization yields a considerable gain with respect to the disjoint optimization  for applications having a low ratio $\eta$, i.e., applications with a high number of bits to be transferred, for a given computational load $w$; ii) the overall energy consumption decreases for computationally intensive applications, i.e., applications characterized by a high $\eta$.

\begin{comment}
\noindent \emph{Example $\#$ 2: On the optimal resources allocation.}
 To grasp how the computational and radio resources are allocated across the users,  in Fig. \ref{figalloc1},  we plot the allocation of the  (normalized) CPU cycles $f_{i_n}/f_T$ and the rates $r_{i_n}$ resulting from Algorithm \ref{alg:Alg_centr}, for a given set of  $\eta_{i_n}=w_{i_n}/b_{i_n}$  (also shown in the figure), and  $f_T=10^{9}$.
 It is interesting to
observe  how the optimal allocation strategy tends to assign
the highest computational rate $f_{i_n}$  to the users with the most demanding applications  (the highest ratio $\eta_{i_n}$).
\begin{figure}[t]
\centering
\includegraphics[width=7.5cm]{fig_alloc.ps}\vspace{-0.3cm}
\caption{Radio and computational resource allocation (SISO channels): $\eta_{i_n}$, $f_{i_n}/f_T$, and $r_{i_n}$ of the users in the cells.}\label{figalloc1}\vspace{-0.4cm}
\end{figure}
\end{comment}

\noindent \emph{Example $\#$ 2: On the convergence speed.}
   To test the convergence speed of Algorithm \ref{alg:Alg_centr}, Fig. \ref{figiterL}  shows the average energy consumption ${E}(\mathbf{Q}^{\nu})$   versus the iteration index $\nu$, for different values of the maximum latency $\tilde{T}_{i_n}$ (assumed to be equal for all users) and different number of receive antennas.  The curves are averaged over $100$ independent channel realizations.
The  interesting result is that  the proposed algorithm converges in very few iterations. Moreover, as expected,
 the energy consumption increases as the delay constraint becomes more stringent because more transmit energy has to be used to respect the latency limit. Finally, it is worth noticing the gain achievable by increasing the number of receive antennas.\\
Since the overall optimization problem is non-convex, the proposed algorithm may fall into a local minimum.
To evaluate this aspect, we ran our algorithm under $1,000$ independent initializations of the initial parameter setting
 $\mathbf{Z}^{0}=(\mathbf{Q}^{0},\mathbf{f}^{0})\in \mathcal{X}$ of Algorithm \ref{alg:Alg_centr} and, quite interestingly, we always ended up with practically the same result, meaning that the differences where within the third decimal point. %Of course, this is not a proof of convergence to the
%absolute minimum, but it says that the method is quite robust to random initialization.

%   To grasp the optimal solution dependence on the random initial point , in Fig. \ref{fighist}
%we plot the estimate of the distribution (probability density function) of the energy $E(\mathbf{Q}^{0})$   corresponding to $1000$ randomly chosen initial (feasible) points $\mathbf{Q}^{0}$ (left plot), and the distribution of the  optimal   energy $E(\mathbf{Q}^{\star}(\mathbf{Z}^{0}))$  achieved by Algorithm 2  (right plot).
%Quite interestingly, our experiments show that the proposed algorithm is robust against random initializations:
%Although the variance of the initial energy is quite large, the optimal final energy tends to be concentrated around a much smaller range of values (that have been observed to differ from each other  on the third decimal). %\vspace{-0.3cm}
\begin{figure}[t]
\centering
\includegraphics[width=7.1cm]{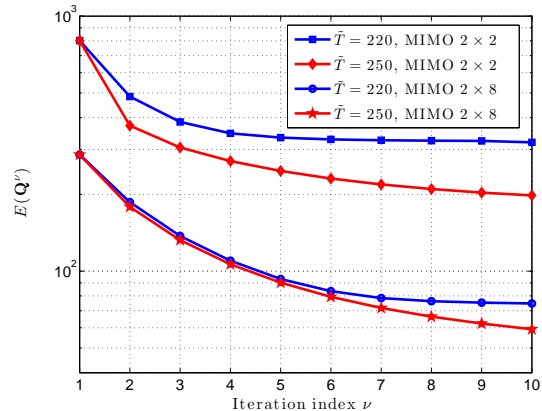}\vspace{-0.2cm}
\caption{Convergence speed: Optimal energy vs. the iteration index  for different
values of $\tilde{T}$.}\label{figiterL}\vspace{-0.4cm}
\end{figure}

%\begin{figure}[h]
% \centering
%        %\begin{subfigure}[b]{0.28\textwidth}
%                \includegraphics[width=1.7in,height=1.5in]{fig_hist2.ps}
%       %\end{subfigure}%
%%        \begin{subfigure}[b]{0.28\textwidth}
%         \hspace{0.00cm}  \includegraphics[width=1.7in,,height=1.5in]{fig_hist1.ps}\vspace{-0.2cm}
%        %\end{subfigure}\vspace{-0.3cm}
%       \caption{Probability density function of the  initial energy
%       $E(\mathbf{Q}^{0})$ (left plot) and of optimal final energy $E(\mathbf{Q}^{\star}(\mathbf{Z}^{0}))$ (right plot).}\label{fighist} \vspace{-0.2cm}
%\end{figure}
%To assess the convergence speed of the proposed algorithm against the number of users, Fig. \ref{fig_numb_it} shows the number of iterations needed for convergence of Algorithm \ref{alg:Alg_centr} versus the number of users, for different values of the ratio $\eta=w_{i_n}/b_{i_n}$, assumed equal for all users. What we observe is that the number of iterations does not increase much with the number of users. Moreover, the algorithm is faster for those applications for which computation offloading is more convenient.
%\begin{figure}[t]
%\centering
%\includegraphics[width=7.1cm]{fig_numb_it.ps}\vspace{-0.2cm}
%\caption{Number of iterations versus the number of active users.}\label{fig_numb_it}\vspace{-0.2cm}
%\end{figure}

\noindent \emph{Example $\#$ 3: Distributed Algorithms.}
Finally, we tested the efficiency of the distributed algorithms proposed in Section \ref{section:decentralized}. We assume $P_{i_n}=P_T=1000$, $\alpha=1$e$-5$ and the termination accuracy $\delta$ is set to $10^{-2}$. Fig. \ref{alg_comp2} shows the energy evolution versus the iteration index $m$, which counts  the overall number of  (inner and outer) iterations %including those needed for attaining the convergence of Step $2$
 in Algorithm \ref{alg:Alg_centr}.
More specifically, we compared three different algorithms used to run Step $2$, namely:   the dual-decomposition method described in Algorithm \ref{alg:Alg_dual_distr}, the  dual-scheme based on the reformulation of the nonconvex problem $\mathcal P$ using slack-variables
   as given in   Algorithm \ref{alg:Alg_dual_distr_y}, and its accelerated version based on the Newton implementation (\ref{Newton_algo}). All implementations are quite fast. As expected, using second order information enhances convergence speed.
  \begin{figure}[h]
\centering
\includegraphics[width=7.5cm]{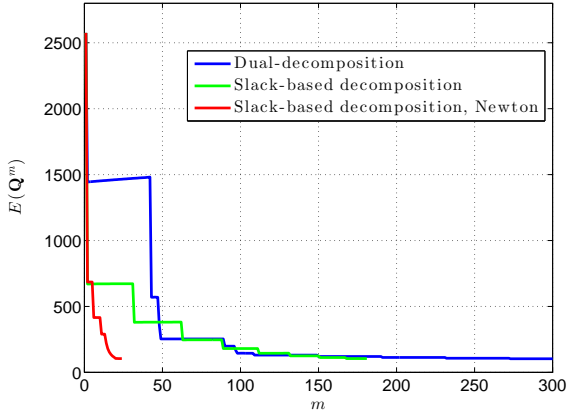}
\caption{Evolution of the global energy  for the  distributed algorithms vs. the iteration index $m$.}\label{alg_comp2}\vspace{-0.4cm}
\end{figure}

\section{Conclusions}
In this paper we formulated the computation offloading problem in a multi-cell mobile edge-computing scenario, where a dense deployment of radio access points facilitates proximity high bandwidth access to computational resources, but increases also intercell interference. We formulated the resource optimization problem as the joint optimization of radio and computational resources, aimed at minimizing MUs' energy consumption, under latency and power budget constraints. In the single-user case, we computed the global optimal solution of the resulting nonconvex optimization problem in closed form. In the more general multi-cell multi-user scenario, we developed centralized and distributed SCA-based algorithms with provable convergence to local optimal solutions of the nonconvex problem. Numerical results show that our algorithms outperform disjoint optimization schemes. Furthermore, the results show, as expected, that offloading is more convenient for applications with high computational load and small number of bits to be exchanged to enable program migration.
 %In this paper we have focused on the single cloud case but a further development will be to extend the
 %%the most appropriate server for offloading according to the  computation availability of the several clouds.
 \vspace{-0.5cm}

\appendix
\subsection{Proof of Theorem \ref{thm:MIMO_SU_SC_energy}}
\label{A:proof Th1} \vspace{-0.1cm}

\noindent (a) It is sufficient to prove the following two facts.\smallskip

\noindent \textbf{Fact  1:} Any stationary point of   the nonconvex problem \ref{P_MIMO_SU_ener}  is a \emph{global} optimal solution of the problem.\smallskip

\noindent \textbf{Fact  2:}  Any stationary point of the \emph{convex} problem \ref{P_MIMO_SU}  (and thus a globally  optimal solution to \ref{P_MIMO_SU}), is also a stationary point of \ref{P_MIMO_SU_ener}, and viceversa.\smallskip

\noindent \textbf{Proof of Fact  1:}  Invoking   \cite[Theorem 3.39 ]{Avriel}, it is sufficient to show that   the objective function $E(\mathbf{Q})$ is  a pseudo-convex function on the convex set $\mathcal X_s$, i.e., \cite[Def. 3.1.3]{Avriel}

 \beq
\forall \bQ,\bY\in \mathcal X_s\,:\,E(\mathbf{Q})< E(\mathbf{Y} ) \; \Rightarrow \; \langle \nabla_{\mathbf{Q}^\ast} E(\mathbf{Y}),\mathbf{Q}-\mathbf{Y} \rangle <0. \label{pseudo-conv}
\eeq

Fix  $\bY\in \mathcal X_s$, and introduce the \emph{convex} $\mathcal C^1$ function $\phi_{\bY}: \mathcal X_s \rightarrow \mathbb{R}$ defined as
\beq\label{def_phi}
\phi_{\bY}(\mathbf{Q})\triangleq \text{tr}(\mathbf{Q})\cdot r(\mathbf{Y})-\text{tr}(\mathbf{Y})\cdot r(\mathbf{Q}).
\eeq
%Note that i)  $\phi_{\bY}(\bY)=0$; and ii)  $\phi_{\bY}(\bullet)$ is convex  on $\mathcal X_s$ and $C^1$, i.e., \beq
%0>\phi_{\bY}(\mathbf{Y})-\phi(\mathbf{Q})\geq \langle \nabla_{\mathbf{Y}}\phi_{\bY}(\mathbf{Q}), \mathbf{Y}-\mathbf{Q} \rangle
%\eeq
 Then, for any $\bQ\in \mathcal{X}_s$ such that $E(\mathbf{Q})< E(\mathbf{Y})$, the following holds:
 \beq
\begin{array}{lll}\label{pseudo_cvx_of_E}
\langle \nabla_{\mathbf{Q}^\ast} E(\mathbf{Y}),\mathbf{Q}-\mathbf{Y}\rangle &\stackrel{(a)}{=}& \dfrac{\langle \nabla_{\mathbf{Q}^\ast} \phi_{\bY}(\mathbf{Y}),\mathbf{Q}-\mathbf{Y}\rangle}{r(\bY)^2} \medskip \\
& \stackrel{(b)}{\leq} & \dfrac{\phi_{\bY}(\mathbf{Q})-\phi_{\bY}(\mathbf{Y})}{r(\bY)^2}\, \stackrel{(c)}{<}\,   0,
\end{array}
\eeq
 where (a) follows from the definition of $\phi_{\bY}$ in \eqref{def_phi}; (b) is due to the convexity of $\phi_{\bY}$ on $\mathcal X_s$; and (c) comes from  $E(\mathbf{Q})< E(\mathbf{Y}) \Rightarrow \phi_{\bY}(\mathbf{Q})< \phi_{\bY}(\mathbf{Y})$. Since (\ref{pseudo_cvx_of_E}) holds for any given  $\bY\in \mathcal X_s$, (\ref{pseudo-conv}) holds true. \hfill $\square$   \smallskip

\noindent \textbf{Proof of Fact  2:} Let us prove the two directions separately.

\noindent $\mathcal{Q}_s \Rightarrow \mathcal{P}_s$: Let $(\bQ^\star, f^\star)$ be the optimal solution of the convex problem $\mathcal{Q}_s$; denote $\tilde{\bQ}^\star\triangleq \mathbf{U}^H \mathbf{Q}^\star \mathbf{U}$. Then, there exist multipliers $\lambda^{\star}_p,\mu^{\star}_p,\alpha^{\star}_p,\boldsymbol{\Phi}^{\star}_p$ such that the tuple $(\tilde{\bQ}^\star,f^\star, \lambda^{\star}_p,\mu^{\star}_p,\alpha^{\star}_p,\beta^{\star}_p,\boldsymbol{\Phi}^{\star}_p)$ satisfies the KKT conditions of $\mathcal Q_s$ (note that  Slater's constraint qualification is satisfied): denoting  $\tilde{r}(\tilde{\bQ}^\star)\triangleq \log_2 |\mathbf{I}+\mathbf{D}^{1/2}\tilde{\mathbf{Q}}^\star  \mathbf{D}^{1/2}  |$, and after some simplifications, one gets
\beq
\begin{array}{llllll}
\text{(a):}\qquad \mathbf{I} -\ds \frac{\mu^{\star}_p}{\log(2)}\, \mathbf{D}^{1/2} (\mathbf{I}+ \mathbf{D}^{1/2} \tilde{\mathbf{Q}}^{\star}\mathbf{D}^{1/2} )^{-1} \mathbf{D}^{1/2} \\  \hfill +\lambda^{\star}_p \,\mathbf{I}-\mathbf{\Phi}^{\star}_p=\mathbf{0}\medskip\\
\text{(b):}\qquad\ds\frac{\mu^{\star}_p\, w \,c}{{f^{\star}}^{2}(\tilde{T}-w/f^{\star})^2}-\alpha^{\star}_p=0\medskip\\
\text{(c):}\qquad 0\leq \lambda^{\star}_p \,\perp\, \left(P_T-\mbox{tr}(\tilde{\mathbf{Q}}^{\star})\right)\geq 0\medskip\\
\text{(d):}\qquad 0< \mu^{\star}_p, \qquad \ds \frac{c}{\tilde{T}-\frac{w}{f^{\star}}}-\tilde{r}(\tilde{\bQ}^\star)= 0\medskip\\
\text{(e):}\qquad\mathbf{0} \preceq\tilde{\mathbf{Q}}^{\star}\,\perp\, \mathbf{\Phi}^{\star}_p\succeq \mathbf{0}\medskip\\
\text{(f):}\qquad 0\leq \alpha^{\star}_p,\qquad    \,f^{\star}=f_T,
 \end{array}\label{KKT_power}\tag{KKT$_{\mathcal Q_s}$}
\eeq
where $\mathbf{A}\perp \textbf{B}$ stands for $\left\langle \mathbf{A},\mathbf{B} \right\rangle =0$, and in (d) and (f) we used the fact that $\mu^{\star}_p$ must be positive and $f^\star=f_T$, respectively  (otherwise \ref{KKT_power} cannot be satisfied).
We prove next that there exist multipliers  $\lambda^{\star}_e,\mu^{\star}_e,\alpha^{\star}_e,\boldsymbol{\Phi}^{\star}_e$  that together with the optimal solution $(\tilde{\bQ}^\star,f^\star)$ of $\mathcal Q_s$ satisfy the KKT conditions of $\mathcal P_s$, i.e.,
\beq
\begin{array}{llllll}
\text{(a$'$):}\qquad \ds \frac{c \cdot \mathbf{I}}{\tilde{r}(\tilde{\bQ}^\star)}  -\,
 \frac{c \cdot \text{tr}(\tilde{\mathbf{Q}}^\star)  \mathbf{D}^{1/2}   (\mathbf{I}+\mathbf{D}^{1/2}\tilde{\mathbf{Q}}^\star \mathbf{D}^{1/2})^{-1}\mathbf{D}^{1/2} }{\tilde{r}(\tilde{\bQ}^\star)^2 \log(2)}\smallskip\\ \hfill -\dfrac{\mu_e^\star}{\log(2)} \mathbf{D}^{1/2} (\mathbf{I}+ \tilde{\mathbf{Q}}^\star \mathbf{D} )^{-1}\mathbf{D}^{1/2}   +\lambda_e^\star \mathbf{I}-\mathbf{\Phi}_e^\star=\mathbf{0}\medskip\\
\text{(b$'$):}\qquad \ds\frac{\mu_e^\star\, w \,c}{f^{\star\,2}(\tilde{T}-{w}/{f^\star})^2}-\alpha_e^\star=0\medskip\\
 \text{(c$'$):}\qquad0\leq \lambda_e^\star\, \perp \,\left(P_T-\mbox{tr}(\tilde{\mathbf{Q}}^\star)\right)\geq 0\medskip\\
\text{(d$'$):}\qquad 0\leq \mu_e^\star\, \perp\, \left(\ds \tilde{r}(\tilde{\bQ}^\star)-\ds \frac{c}{\tilde{T}-{w}/{f^\star}}\right)\geq 0\medskip\\
\text{(e$'$):}\qquad\mathbf{0}\preceq \tilde{\mathbf{Q}}\,\perp \, \mathbf{\Phi}_e^\star\succeq \mathbf{0}\medskip\\
\text{(f$'$):}\qquad 0\leq \alpha_e^\star \, \perp\, (f_T-f^\star)\geq 0.
 \end{array}\label{KKT}\tag{KKT$_{\mathcal P_s}$}\eeq
%\noindent \textbf{{Equation (a$'$)}:}
 Plugging (a) of (\ref{KKT_power})  in (a$'$) of (\ref{KKT}) and using the fact that $\mu_p^\star>0$, we obtain:
\beq \label{lambda1}
\begin{array}{lll}
\lambda^{\star}_e\, \mathbf{I}=-\ds \frac{c\, \mathbf{I}}{\tilde{r}(\tilde{\mathbf{Q}}^{\star})}+  \ds \dfrac{(1+\lambda^{\star}_p)}{\mu^{\star}_p} \left( \ds \dfrac{c \,\text{tr}(\tilde{\mathbf{Q}}^{\star})}{\tilde{r}(\tilde{\mathbf{Q}}^{\star})^2}+\ds \mu^{\star}_e\right)\cdot \mathbf{I} \medskip\\
\qquad\quad+\,\mathbf{\Phi}^{\star}_e- \dfrac{1}{\mu^{\star}_p}\,\left( \ds \frac{c\, \text{tr}(\tilde{\mathbf{Q}}^{\star})}{\tilde{r}(\tilde{\mathbf{Q}}^{\star})^2}+\mu^{\star}_e\right)\cdot \mathbf{\Phi}^{\star}_p,
\end{array}
\eeq
which is satisfied if one set   $\mathbf{\Phi}^{\star}_e$, $\lambda^{\star}_e$, and $\mu^{\star}_e$ to
\beq
\begin{array}{lll}
\mathbf{\Phi}^{\star}_e &\triangleq &  \dfrac{1}{\mu^{\star}_p}\left( \ds \frac{c\, \text{tr}(\tilde{\mathbf{Q}}^{\star})}{\tilde{r}(\tilde{\mathbf{Q}}^{\star})^2}+\mu^{\star}_e\right)\cdot \mathbf{\Phi}^{\star}_p\medskip\\
{\mu}^{\star}_e &\triangleq &  \ds \frac{c\, \mu^{\star}_p}{\tilde{r}(\tilde{\mathbf{Q}}^{\star})(1+\lambda^{\star}_p)}-\ds \frac{c\, \text{tr}(\tilde{\mathbf{Q}}^{\star})}{\tilde{r}(\tilde{\mathbf{Q}}^{\star})^2}\medskip\\
\lambda^{\star}_e &\triangleq & 0.
\end{array}\label{cond_mult}
\eeq
By (b$'$) it must be
\beq
\alpha^{\star}_e = \ds\frac{\mu_e^\star\, w \,c}{f^{\star\,2}(\tilde{T}-{w}/{f^\star})^2}.
\label{mu_alpha_cond}
\eeq
Note that, to be a valid candidate solution of   \ref{KKT}, $\mu^{\star}_e$ must be nonnegative [cf. (d$'$)], which by (\ref{cond_mult}), is equivalent to
\beq \label{kappa1}
     \dfrac{1+\lambda_p^\star}{\mu^{\star}_p}\cdot {\text{tr}(\tilde{\mathbf{Q}}^{\star})} \leq \tilde{r}(\tilde{\mathbf{Q}}^{\star}).
\eeq
We show next that \eqref{kappa1} holds true.  By multiplying both sides of (a) by    $\tilde{\mathbf{Q}}^{\star}$ and using the complementarity condition $\langle\mathbf{\Phi}^{\star}_p, \tilde{\mathbf{Q}}^{\star}\rangle=0$  [cf. (e)] we get \vspace{-0.2cm}
\beq \label{cond1}
\begin{split}
 \dfrac{1+\lambda_p^\star}{\mu^{\star}_p}\!\cdot {\text{tr}(\tilde{\mathbf{Q}}^{\star})}&\! =\!\ds \frac{1}{\log(2)}\langle{\tilde{\mathbf{Q}}^{\star}}, \mathbf{D}^{1/2}   (\mathbf{I}\!+\mathbf{D}^{1/2} \tilde{\mathbf{Q}}^{\star} \mathbf{D}^{1/2})^{-1}\mathbf{D}^{1/2}\rangle \\ & =\langle \nabla_{\mathbf{Q}^{*}}\tilde{r}(\tilde{\mathbf{Q}}^{\star}),\tilde{\mathbf{Q}}^{\star}\rangle \leq  \tilde{r}(\tilde{\mathbf{Q}}^{\star}),
\end{split}
\eeq
where in the last inequality we used the concavity of the rate function $\tilde{r}(\bullet)$, i.e.,
\beq \label{first_order}
\tilde{r}(\mathbf{Y})\leq \tilde{r}(\mathbf{W})+\langle \nabla_{\mathbf{Q}^{*}}\tilde{r}(\mathbf{W}),\mathbf{Y}-\mathbf{W}\rangle, \quad \forall \mathbf{Y},\mathbf{W}\succeq \mathbf{0}
\eeq
evaluated at  $\mathbf{Y}=\mathbf{0}$ and $\mathbf{W}=\tilde{\mathbf{Q}}^{\star}$.
The desired result, $\mu_e^\star \geq 0$, follows readily combining (\ref{kappa1}) and (\ref{cond1}).

%Note that there always exist such $\mu^{\star}_e$ and $\alpha^{\star}_e$.
We show now  that the obtained  tuple  $(\tilde{\bQ}^\star,f^\star,\lambda^{\star}_e,\mu^{\star}_e,\alpha^{\star}_e,$ $\boldsymbol{\Phi}^{\star}_e)$ satisfies \ref{KKT}. Indeed,   (a$'$) follows from \eqref{cond_mult};  given  $\mu_e^\star \geq 0$, (b$'$) is satisfied by   $\alpha^{\star}_e$  as  in \eqref{mu_alpha_cond};   (c$'$) follows from  $P_T- \text{tr}(\tilde{\mathbf{Q}}^{\star})\geq 0$ [cf. (c)] and $\lambda^{\star}_e=0$;  (d$'$) follows from  $\mu^{\star}_e\geq 0$ and the second equality in (d). Finally,  it is not difficult to see that  $\mathbf{\Phi}^{\star}_e$ given by (\ref{cond_mult}) satisfies   (e$'$); and finally (f$'$) is trivially met by  $\alpha^{\star}_e\geq 0$ in \eqref{mu_alpha_cond}.
 This completes the first part of the proof.

  \noindent $\mathcal{P}_s \Rightarrow \mathcal{Q}_s$: the proof follows the same idea as for $\mathcal{Q}_s \Rightarrow \mathcal{P}_s$; we then only sketch  the main steps.  Let $(\tilde{\bQ}^\star,f^\star, \lambda^{\star}_e,\mu^{\star}_e,\alpha^{\star}_e, \boldsymbol{\Phi}^{\star}_e)$ be a tuple satisfying \ref{KKT} (whose existence is guaranteed by the  Slater's constraint qualification). We prove next that there exist multipliers $(\lambda^{\star}_p,\mu^{\star}_p,\alpha^{\star}_p, \boldsymbol{\Phi}^{\star}_p)$ such that $(\tilde{\bQ}^\star,f^\star, \lambda^{\star}_p,\mu^{\star}_p,\alpha^{\star}_p, \boldsymbol{\Phi}^{\star}_p)$ satisfies \ref{KKT_power}.
Define  $$\kappa_e=\mu^{\star}_e+\ds \frac{c \,\text{tr}(\tilde{\mathbf{Q}}^\star)}{\tilde{r}(\tilde{\bQ}^\star)^2}> 0.$$ %
Given (a$'$), it can be easily seen that (a) is satisfied if $\mathbf{\Phi}^{\star}_p$,  $\lambda^{\star}_p$, and ${\mu}^{\star}_p$ are chosen as\beq \label{cond_pow}
\mathbf{\Phi}^{\star}_p=\ds \frac{\mu^{\star}_p}{\kappa_e}\, \mathbf{\Phi}^{\star}_e,\quad
{\mu}^{\star}_p  =    \ds \frac{\kappa_e}{\lambda^{\star}_e+\ds \frac{c}{\tilde{r}(\tilde{\bQ}^\star)}},\quad\text{and}\quad
\lambda^{\star}_p =0.
\eeq
From  (b) it must also be
\beq
\alpha^{\star}_p  =  \ds\frac{\mu_p^\star\, w \,c}{f^{\star\,2}(\tilde{T}-{w}/{f^\star})^2}.
\label{mu_alpha_cond2}
\eeq

 It is not difficult to check that the obtained tuple    $(\tilde{\bQ},f^\star,$ $\lambda^{\star}_p,\mu^{\star}_p,\alpha^{\star}_p,\boldsymbol{\Phi}^{\star}_p)$ satisfies (a), (b), (c), (e), and (f) of \ref{KKT_power};   the only condition that needs a proof is the equality constraint in (d), as given next.
%(a) follows from (\ref{cond_pow}); (b) is satisfied by any   $\mu^{\star}_p>0$ and  $\alpha^{\star}_p>0$ chosen as in \eqref{mu_alpha_cond2};  (c) holds true thanks to $\text{tr}(\tilde{\mathbf{Q}}^{\star})=P_T$ and $\lambda^{\star}_p\geq 0$, the latter due to   $\mu^{\star}_p$  satisfying (\ref{mu_alpha_cond2}).
% Some further observations requires the proof of (d).

Suppose by contradiction that $\ds \tilde{r}(\tilde{\bQ}^\star)- \ds \frac{c}{\tilde{T}-{w}/{f^\star}}>0$. Then, it follows from  (d$'$)
 that  $\mu^{\star}_e=0$, and (a$'$) reduces to
 \beq\nonumber
\ds \frac{c \, \mathbf{I}}{\tilde{r}(\tilde{\bQ}^\star)}  -\,
 \ds \frac{c \, \text{tr}(\tilde{\mathbf{Q}}^\star)  \mathbf{D}^{1/2}   (\mathbf{I}+\mathbf{D}^{1/2}\tilde{\mathbf{Q}}^\star \mathbf{D}^{1/2})^{-1} \mathbf{D}^{1/2} }{\log(2) \tilde{r}(\tilde{\bQ}^\star)^2}= \!\!-\lambda_e^\star \mathbf{I}+\mathbf{\Phi}_e^\star.
\eeq
Multiplying the above  equation by $\tilde{\mathbf{Q}}^\star$ and using the complementary condition (e$'$), we  get
\beq
\lambda_e^\star=\ds \frac{c}{\tilde{r}(\tilde{\bQ}^{*})^2}\left(\langle \nabla_{\mathbf{Q}^{*}}\tilde{r}(\tilde{\mathbf{Q}}^{\star}),\tilde{\mathbf{Q}}^{\star}\rangle -r(\tilde{\bQ}^{\star})\right),
\eeq
which, given $\lambda_e^\star\geq 0$ [cf. (c$'$)] and   $\langle \nabla_{\mathbf{Q}^{*}}\tilde{r}(\tilde{\mathbf{Q}}^{\star}),\tilde{\mathbf{Q}}^{\star}\rangle\leq \tilde{r}(\tilde{\mathbf{Q}}^{\star})$ [due to (\ref{first_order})], can be satisfied only if
  $
\langle \nabla_{\mathbf{Q}^{*}}\tilde{r}(\tilde{\mathbf{Q}}^{\star}),\tilde{\mathbf{Q}}^{\star}\rangle =r(\tilde{\bQ}^\star)
  $, i.e.,
 \beq
 \begin{split}\nonumber %\label{rate_ineq1}
 &\log_2\det(\mathbf{I}+\mathbf{D}^{1/2}\tilde{\mathbf{Q}^\star}\mathbf{D}^{1/2}) \\&\quad=
 \text{tr}\left(\tilde{\mathbf{Q}}^\star \mathbf{D}^{1/2}(\mathbf{I}+\mathbf{D}^{1/2}\tilde{\mathbf{Q}}^\star\mathbf{D}^{1/2})^{-1}  \cdot \mathbf{D}^{1/2}\right)\cdot \ds \frac{1}{\log(2)}.
 \end{split}
 \eeq
 % We show  now that \eqref{rate_ineq1} holds true  if and only if $\tilde{\mathbf{Q}}^\star=\mathbf{0}$, which contradicts $\text{tr}(\tilde{\mathbf{Q}}^{\star})=P_T$.
  Denoting by   $(\sigma_i=\sigma_i(\mathbf{D}^{1/2}\tilde{\mathbf{Q}}^\star\mathbf{D}^{1/2}))_{i=1}^r\geq \mathbf{0}$ the non-negative eigenvalues of $\mathbf{D}^{1/2}\tilde{\mathbf{Q}}^\star\mathbf{D}^{1/2}$, the above equality can be rewritten as
   \beq\nonumber
  \ds \sum_{i=1}^{r} \log(1+\sigma_i)= \ds \sum_{i=1}^{r}\ds \frac{\sigma_i}{1+\sigma_i},
   \eeq
which can be true only if    $\sigma_i=0$ for all $i=1,\cdots,r$, and thus
 $\tilde{\mathbf{Q}}^\star=\mathbf{0}$ (note that $\mathbf{D}\neq\mathbf{0}$). This however    is in contradiction with the fact that $\mathbf{Q}^{\star}$ is an optimal solution of $\mathcal Q_s$.

 % Additionally, it is not difficult to see that  $\mathbf{\Phi}^{\star}_p$ given by (\ref{cond_pow}) satisfies   (e) using (e$'$). Finally (f) is trivially met by  $\alpha^{\star}_p> 0$ in \eqref{mu_alpha_cond2},  observing that $\mu^{\star}_e>0$ implies from (b$'$) $\alpha^{\star}_e>0$ and then $f^{\star}=f_T$ from
 % (f$'$).

\noindent (b): Invoking part  (a) of the theorem,  the solution   $(\bQ^\star, f^\star)$  of
\ref{P_MIMO_SU} (and thus \ref{P_MIMO_SU_ener}) can be computed  solving \ref{KKT_power}. Denote  $\tilde{\bQ}^\star\triangleq \mathbf{U}^H \mathbf{Q}^\star \mathbf{U}$.
Multiplying (a) of \ref{KKT_power} by $\tilde{\mathbf{Q}}^{\star}$ and using (e), we get
\beq \label{grad_qs}
\mathbf{I} -\alpha  \mathbf{D}^{1/2} (\mathbf{I}+ \mathbf{D}^{1/2} \tilde{\mathbf{Q}}^{\star}\mathbf{D}^{1/2} )^{-1} \mathbf{D}^{1/2}=\mathbf{0}
\eeq
with $\alpha\triangleq  {\mu^{\star}_p}/\log(2)$ (recall that one can set $\lambda^{\star}_p=0$).
By solving (\ref{grad_qs}) and using $\tilde{\bQ}^\star\triangleq \mathbf{U}^H \mathbf{Q}^\star \mathbf{U}$ one obtains the desired expression of $\mathbf{Q}^\star $ as in (\ref{MIM0_SU}). Moreover, it follows  from  (f) that   $f^{\star}=f_T$.
The only thing left to show is  how to compute  $\alpha$ (and thus $\mu^\star_p$) efficiently. Using the optimal structure of $\mathbf{Q}^\star$ and denoting $r_e\triangleq \text{rank}(\mathbf{Q}^\star)$, conditions (c) and (d) reduce respectively to
\beq \label{cc_cond}
 \alpha=2^{\ds \frac{c}{r_e L}-\ds \frac{1}{r_e} \sum_{i=1}^{r_e} \log_2(d_i)}\,\,\text{and}\,\,\,\, \sum_{i=1}^{r_e}  \left(\alpha-\ds \frac{1}{d_i}\right)\leq P_T,
\eeq
with $L=\tilde{T}-\frac{w}{f_T}$. Note that Slater's constraint qualification guarantees that there exist $\alpha$ and $r_e$ satisfying  $\eqref{cc_cond}$. Moreover, it is not difficult to check that they can be efficiently computed using the  procedure described in Algorithm \ref{algorithm:Alg_SU_SC_MIMO}.

\vspace{-0.6cm}
\bibliographystyle{IEEEtran}
%\balance
\bibliography{scutari_refs}

\begin{thebibliography}{10}
\providecommand{\url}[1]{#1}
\csname url@samestyle\endcsname
\providecommand{\newblock}{\relax}
\providecommand{\bibinfo}[2]{#2}
\providecommand{\BIBentrySTDinterwordspacing}{\spaceskip=0pt\relax}
\providecommand{\BIBentryALTinterwordstretchfactor}{4}
\providecommand{\BIBentryALTinterwordspacing}{\spaceskip=\fontdimen2\font plus
\BIBentryALTinterwordstretchfactor\fontdimen3\font minus
  \fontdimen4\font\relax}
\providecommand{\BIBforeignlanguage}[2]{{%
\expandafter\ifx\csname l@#1\endcsname\relax
\typeout{** WARNING: IEEEtran.bst: No hyphenation pattern has been}%
\typeout{** loaded for the language `#1'. Using the pattern for}%
\typeout{** the default language instead.}%
\else
\language=\csname l@#1\endcsname
\fi
#2}}
\providecommand{\BIBdecl}{\relax}
\BIBdecl

\bibitem{SarScuBarSPAWC14}
S.~Sardellitti, G.~Scutari, and S.~Barbarossa, ``Joint optimization of radio
  and computational resources for multicell mimo mobile cloud computing,'' in
  \emph{Proc. of 2014 IEEE International Workshop on Signal Processing Advances
  in Wireless Communications (SPAWC '14)}, Toronto, Canada, Jun 22-25 2014.

\bibitem{Sard_Cloud_net}
------, ``Distributed joint optimization of radio and computational resources
  for mobile cloud computing,'' in \emph{Proc. of IEEE Int. Conf. on Cloud
  Networking (Cloudnet 14)}, Luxembourg, 8--10 October 2014.

\bibitem{Palacin}
M.~Palacin, ``Recent advances in rechargeable battery materials: a chemists
  perspective,'' \emph{Chem. Soc. Rev.}, vol.~38, no.~9, pp. 2565--2575, Sept.
  2009.

\bibitem{Sharifi-Kafaie-Kashefi}
M.~Sharifi, S.~Kafaie, and O.~Kashefi, ``{A Survey and Taxonomy of Cyber
  Foraging of Mobile Devices},'' \emph{IEEE Communications Surveys and
  Tutorials}, vol.~14, pp. 1232--1243, Fourth Quarter 2012.

\bibitem{Kumar-Liu-Lu-Bhargava}
K.~Kumar, J.~Liu, Y.-H. Lu, and B.~Bhargava, ``A survey of computation
  offloading for mobile systems,'' \emph{Mobile Networks and Applications},
  vol.~18, no.~1, pp. 129--140, February 2013.

\bibitem{Fernando_et_al}
N.~Fernando, S.~Loke, and W.~Rahayu, ``Mobile cloud computing: A survey,''
  \emph{Future Generation Computer Systems}, vol.~29, pp. 84--106, January
  2013.

\bibitem{Yang}
K.~Yang, S.~Ou, and H.~Chen, ``On effective offloading services for
  resource-constrained mobile devices running heavier mobile internet
  applications,'' \emph{Communications Magazine, IEEE}, vol.~46, no.~1, pp.
  56--63, 2008.

\bibitem{Wolski}
R.~Wolski, S.~Gurun, C.~Krintz, and D.~Nurmi, ``Using bandwidth data to make
  computation offloading decisions,'' \emph{IEEE International Symposium on
  Parallel and Distributed Processing, IPDPS}, pp. 1--8, 2008.

\bibitem{Zhang}
X.~Zhang, A.~Kunjithapatham, S.~Jeong, and S.~Gibbs, ``Towards an elastic
  application model for augmenting the computing capabilities of mobile devices
  with cloud computing,'' \emph{Journal Mobile Networks and Application},
  vol.~16, no.~3, pp. 270--284, June.

\bibitem{Kaushi}
N.~Kaushi and J.~Kumar, ``A computation offloading framework to optimize energy
  utilisation in mobile cloud computing environment,'' \emph{Int. Journal of
  Computer Applications and Information Technology}, vol.~5, no.~2, pp. 61--69,
  April-May.

\bibitem{Cardellini}
\BIBentryALTinterwordspacing
V.~Cardellini, V.~D.~N. Persone, V.~D. Valerio, F.~Facchinei, V.~Grassi, F.~L.
  Presti, and V.~Piccialli, ``A game-theoretic approach to computation
  offloading in mobile cloud computing,'' \emph{Mathematical Programming},
  (submitted, August 2013). [Online]. Available:
  \url{http://www.optimization-online.org/DB_HTML/2013/08/3981.html}
\BIBentrySTDinterwordspacing

\bibitem{Kumar_Lu}
K.~Kumar and Y.-H. Lu, ``Cloud computing for mobile users: can offloading
  computation save energy?'' \emph{Computer Society, IEEE}, vol.~43, no.~4, pp.
  51--56, April 2010.

\bibitem{Cui}
Y.~Cui, X.~Ma, H.~Wang, I.~Stojmenovic, and J.~Liu, ``A survey of energy
  efficient wireless transmission and modeling in mobile cloud computing,''
  \emph{Mobile Networks and Applications}, vol.~18, no.~1, pp. 148--155,
  February 2013.

\bibitem{Miettinen}
A.~P. Miettinen and J.~K. Nurminen, ``Energy efficiency of mobile clients in
  cloud computing,'' \emph{in Proc. of the 2nd USENIX conference on Hot topics
  in cloud computing}, pp. 4--4, June 2010.

\bibitem{Huang-Wang-Niyato_TW12}
D.~Huang, P.~Wang, and D.~Niyato, ``A dynamic offloading algorithm for mobile
  computing,'' \emph{{IEEE} Trans. Wireless Commun.}, vol.~11, no.~6, pp.
  1991--1995, April 2012.

\bibitem{Kao:Globecom:2014}
B.~K. Yi-Hsuan~Kao, ``Optimizing mobile computational offloading with delay
  constraints,'' in \emph{Proc. of Global Communication Conference (Globecom
  14)}, Austin, TX, USA, Dec. 8-12 2014, pp. 49--62.

\bibitem{liu2013gearing}
F.~Liu, P.~Shu, H.~Jin, L.~Ding, J.~Yu, D.~Niu, and B.~Li, ``Gearing
  resource-poor mobile devices with powerful clouds: architectures, challenges,
  and applications,'' \emph{Wireless Communications, IEEE}, vol.~20, no.~3,
  2013.

\bibitem{Sanaei-14}
Z.~Sanaei, S.~Abolfazli, A.~Gani, and R.~Buyya, ``Heterogeneity in mobile cloud
  computing: Taxonomy and open challenges,'' \emph{IEEE Communications Surveys
  and Tutorials}, vol.~20, no.~3, pp. 369--392, 2014.

\bibitem{Maui}
E.~Cuervo, A.~Balasubramanian, D.~Cho, A.~Wolman, S.~Saroiu, R.~Chandra, and
  P.~Bahl, ``Maui: making smartphones last longer with code offload,'' in
  \emph{Proc. of the ACM International Conference on Mobile Systems,
  Applications, and Services}, San Francisco, CA, USA, 15--18 June 2010, pp.
  49--62.

\bibitem{Kosta_2012}
S.~Kosta, A.~Aucinas, P.~Hui, R.~Mortier, and X.~Zhang, ``Thinkair: Dynamic
  resource allocation and parallel execution in the cloud for mobile code
  offloading,'' in \emph{Proc. of INFOCOM 2012}, March 2012, pp. 945--953.

\bibitem{Phone2Cloud}
F.~Xia, F.~Ding, J.~Li, X.~Kong, L.~Yang, and J.~Ma, ``Phone2cloud: Exploiting
  computation offloading for energy saving on smartphones in mobile cloud
  computing,'' \emph{Information Systems Frontiers}, vol.~16, pp. 95--111,
  2014.

\bibitem{Wen}
W.~Zhang, Y.~Wen, K.~Guan, D.~Kilper, H.~Luo, and D.~O. Wu, ``Energy-optimal
  mobile cloud computing under stochastic wireless channel,'' \emph{IEEE
  Transaction on Wireless Communications}, vol.~12, no.~9, pp. 2716--2720,
  September.

\bibitem{Barbarossa_FuNeMS2013}
S.~Barbarossa, S.~Sardellitti, and P.~{Di Lorenzo}, ``Computation offloading
  for mobile cloud computing based on wide cross-layer optimization,'' in
  \emph{Proc. of Future Network and Mobile Summit (FuNeMS2013)}, Lisboa,
  Portugal, 3--5 July 2013, pp. 1--10.

\bibitem{Barbarossa_SPAWC2013}
------, ``Joint allocation of computation and communication resources in
  multiuser mobile cloud computing,'' in \emph{Proc. of IEEE Workshop on Signal
  Processing Advances in Wireless Communications (SPAWC 2013)}, Darmstadt,
  Germany, 16--19 June 2013, pp. 26--30.

\bibitem{Barbarossa_SPM2014}
------, ``Communicating while computing: Distributed mobile cloud computing
  over 5{G} heterogeneous networks,'' \emph{IEEE Signal Processing Magazine},
  vol.~31, no.~6, pp. 45--55, November.

\bibitem{Chen_2014}
X.~Chen, ``Decentralized computation offloading game for mobile cloud
  computing,'' \emph{IEEE Transactions on Parallel and Distributed Systems},
  (to appear) 2014.

\bibitem{Satyanarayanan-Bahl-Caceres-Davies}
M.~Satyanarayanan, P.~Bahl, R.~Caceres, and N.~Davies, ``The case for vm-based
  cloudlets in mobile computing,'' \emph{IEEE Pervasive Computing}, vol.~8,
  no.~4, pp. 14--23, Oct.-Dec. 2009.

\bibitem{TROPIC}
\BIBentryALTinterwordspacing
TROPIC, ``Distributed computing, storage and radio resource allocation over
  cooperative femtocells.'' [Online]. Available:
  \url{http://www.ict-tropic.eu.}
\BIBentrySTDinterwordspacing

\bibitem{ETSI-MEC}
\BIBentryALTinterwordspacing
ETSI, ``Etsi first meeting of new standardization group on mobile-edge
  computing.'' [Online]. Available:
  \url{http://www.etsi.org/news-events/news/838-2014-10-news-etsi-announces-first-meeting-of-new-standardization-group-on-mobile-edge-computing.}
\BIBentrySTDinterwordspacing

\bibitem{Scutari_ICASSP14}
G.~Scutari, F.~Facchinei, L.~Lampariello, and P.~Song, ``Parallel and
  distributed methods for nonconvex optimization,'' in \emph{Proc. of IEEE Int.
  Conf. on Acoustics, Speech, and Signal Process. (ICASSP 14)}, Florence,
  Italy, 4--9 May 2014.

\bibitem{Scutari_nonconvex}
\BIBentryALTinterwordspacing
------, ``{Distributed Methods for Nonconvex Constraints Multi-Agent
  Problems$-${P}art I: {T}heory},'' \emph{IEEE Trans. on Signal Processing},
  (submitted, Oct. 2014). [Online]. Available:
  \url{http://arxiv.org/abs/1410.4754}
\BIBentrySTDinterwordspacing

\bibitem{scutari_facchinei_et_al_tsp13}
G.~Scutari, F.~Facchinei, P.~Song, D.~Palomar, and J.-S. Pang, ``Decomposition
  by partial linearization: {P}arallel optimization of multi-agent systems,''
  \emph{IEEE Trans.\ on Signal Process.}, vol.~62, pp. 641--656, Feb. 2014.

\bibitem{ThiTao05}
H.~A.~L. Thi and P.~Tao, ``The dc programming and dca revised with dc models of
  real world nonconvex optimization problems,'' \emph{Annals of operations
  research}, vol. 133, no. 1-4, pp. 23--46, January 2005.

\bibitem{AlvaradoScutariPange14}
A.~Alvarado, G.~Scutari, and J.~Pang, ``{A new decomposition method for
  multiuser DC- programming and its applications to Physical Layer Security},''
  \emph{IEEE Trans.\ on Signal Process.}, vol.~62, no.~11, pp. 2984--2998, June
  2014.

\bibitem{HunterLange04}
D.~R. Hunter and K.~Lange, ``{A tutorial on MM algorithms},'' \emph{Amer.
  Statist.}, vol.~58, no.~1, pp. 30--37, 2004.

\bibitem{Sriperumbudur09}
B.~K. Sriperumbudur and G.~R.~G. Lanckriet, ``{A proof of convergence of the
  concave-convex procedure using Zangwill's theory},'' \emph{Neural
  Computation}, vol.~21, pp. 1391--1407, June 2012.

\bibitem{Boyd13}
\BIBentryALTinterwordspacing
S.~Boyd, ``{Sequential Convex Programming},'' \emph{Lecture Note}, 2013.
  [Online]. Available:
  \url{http://www.stanford.edu/class/ee364b/lectures/seq_slides.pdf}
\BIBentrySTDinterwordspacing

\bibitem{razaviyayn2013unified}
M.~Razaviyayn, M.~Hong, and Z.-Q. Luo, ``A unified convergence analysis of
  block successive minimization methods for nonsmooth optimization,''
  \emph{SIAM J. on Opt.}, vol.~23, no.~2, pp. 1126--1153, 2013.

\bibitem{scutari_facchinei_et_al_tsp14}
\BIBentryALTinterwordspacing
F.~Facchinei, G.~Scutari, and S.~Sagratella, ``Parallel selective algorithms
  for big data optimization,'' \emph{IEEE Trans.\ on Signal Process.}, (to
  appear), 2014. [Online]. Available: \url{http://arxiv.org/abs/1402.5521}
\BIBentrySTDinterwordspacing

\bibitem{patriksson1993partial}
M.~Patriksson, ``Partial linearization methods in nonlinear programming,''
  \emph{Journal of Optimization Theory and Applications}, vol.~78, no.~2, pp.
  227--246, August 1993.

\bibitem{BolteShohamTeboulle13}
J.~Bolte, S.~Shoham, and M.~Teboulle, ``{Proximal alternating linearized
  minimization for nonconvex and nonsmooth problems},'' \emph{Mathematical
  Programming}, vol. 146, no. 1-2, pp. 459--494, August 2014.

\bibitem{Facchinei-Pang_FVI03}
F.~Facchinei and J.-S. Pang, \emph{{Finite-Dimensional Variational Inequalities
  and Complementarity Problem}}.\hskip 1em plus 0.5em minus 0.4em\relax
  Springer-Verlag, New York, 2003.

\bibitem{Avriel}
M.~Avriel, W.~E. Diewert, S.~Schaible, and I.~Zang, \emph{Generalized
  Concavity}, ser. Classics in Applied Mathematics (Book 63).\hskip 1em plus
  0.5em minus 0.4em\relax SIAM--Society for Industrial $\&$ Applied
  Mathematics, 2010.

\end{thebibliography}

\end{document}